\documentclass[runningheads]{llncs}
\usepackage[utf8]{inputenc}
\usepackage{graphicx}
\usepackage{amsmath,amsfonts,amssymb}
\usepackage{bm}
\usepackage{booktabs}
\usepackage{hyperref}
\hypersetup{%
  colorlinks=true,
}

\DeclareMathOperator*{\argmin}{arg\,min}
\DeclareMathOperator{\Forall}{\forall}

\begin{document}
\title{Uncertainty Estimation in Medical Image Denoising with Bayesian Deep Image Prior}
\titlerunning{Uncertainty Estimation with Bayesian Deep Image Prior}
%
\author{Max-Heinrich Laves \and
Malte Tölle \and
Tobias Ortmaier}

\authorrunning{MH.\ Laves et al.}

\institute{Leibniz Universität Hannover, Hanover, Germany\\
\email{\{lastname\}@imes.uni-hannover.de}}
\maketitle              
\begin{abstract}
Uncertainty quantification in inverse medical imaging tasks with deep learning has received little attention.
However, deep models trained on large data sets tend to hallucinate and create artifacts in the reconstructed output that are not anatomically present.
We use a randomly initialized convolutional network as parameterization of the reconstructed image and perform gradient descent to match the observation, which is known as deep image prior.
In this case, the reconstruction does not suffer from hallucinations as no prior training is performed.
We extend this to a Bayesian approach with Monte Carlo dropout to quantify both aleatoric and epistemic uncertainty.
The presented method is evaluated on the task of denoising different medical imaging modalities.
The experimental results show that our approach yields well-calibrated uncertainty.
That is, the predictive uncertainty correlates with the predictive error.
This allows for reliable uncertainty estimates and can tackle the problem of hallucinations and artifacts in inverse medical imaging tasks.

\keywords{Variational inference, hallucination, deep learning}
\end{abstract}
\section{Introduction}

Noise in medical imaging affects all modalities, including X-ray, magnetic resonance imaging (MRI), computed tomography (CT), ultrasound (US) or optical coherence tomography (OCT) and can obstruct important details for medical diagnosis \cite{Gondara2016,Agostinelli2013,Laves2019ECBO}.
Besides ``classical'' approaches with linear and non-linear filters, such as the Wiener filter, or wavelet-denoising \cite{Chang2000,Rabbani2009}, convolutional neural networks (CNN) have proven to yield superior performance in denoising of natural and medical images \cite{Zhang2017,Laves2019ECBO}.

The task of denoising is an inverse image problem and aims at reconstructing a clean image $ \hat{\bm{x}} $ from a noisy observation $ \tilde{\bm{x}} = \bm{c} \circ \bm{x} $.
A common assumption of the noise model $ \bm{c} $ of the image $ \tilde{\bm{x}} $ is additive white Gaussian noise with zero mean and standard deviation $ \sigma $ \cite{Salinas2007,Zhang2017}.
Given a noisy image $ \tilde{\bm{x}} $, the denoising can be expressed as optimization problem of the form
\begin{equation}
    \hat{\bm{x}} = \argmin \Big\{ \mathcal{L}(\tilde{\bm{x}}, \hat{\bm{x}}) + \lambda \mathcal{R}(\hat{\bm{x}}) \Big\} ~ .
\end{equation}
The reconstruction $ \hat{\bm{x}} $ should be close to $ \tilde{\bm{x}} $ by means of a similarity metric $ \mathcal{L} $, but with substantially less noise.
The regularizer $ \mathcal{R} $ expresses a prior on the reconstructed images, which leads to $ \hat{\bm{x}} $ having less noise than $ \tilde{\bm{x}} $.
One usually imposes a smoothness constrain by penalizing first or higher order spatial derivatives of the image \cite{Sotiras2013}.
More recently, denoising autoencoders have successfully been used to implicitly learn a regularization prior from a data set with corrupted and uncorrupted data samples \cite{Jain2009}.
Autoencoders are usually composed of an encoding and decoding part with a data bottleneck in between.
The encoder extracts important visual features from the noisy input image and the decoder reconstructs the input from the extracted features using learned image statistics.

This, however, creates the root problem of medical image denoising with deep learning that is addressed in this paper.
The reconstruction is in accordance with the expectation of the denoising autoencoder based on previously learned information.
At worst, the reconstruction can contain false image features, that look like valid features, but are not actually present in the input image.
Due to the excellent denoising performance of autoencoders, those false features can be indistinguishable from valid features to a layperson and are embedded in an otherwise visually appealing image.
This phenomenon is known as \emph{hallucination} and, while acceptable in the reconstruction of natural images \cite{Wang2014}, must be avoided at all costs in medical imaging (see Fig.\,\ref{fig:opener}).
Hallucinations can lead to false diagnoses and thus severely compromise patient safety.

\begin{figure}[t]
    \centering
    \parbox{4.2cm}{\centering
    \includegraphics[width=4cm]{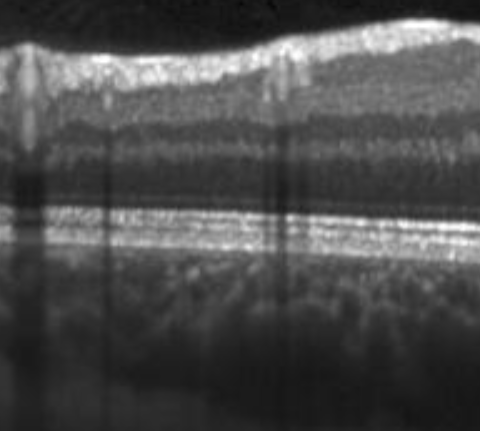} \\
    ground truth} ~
    \parbox{4.2cm}{\centering
    \includegraphics[width=4cm]{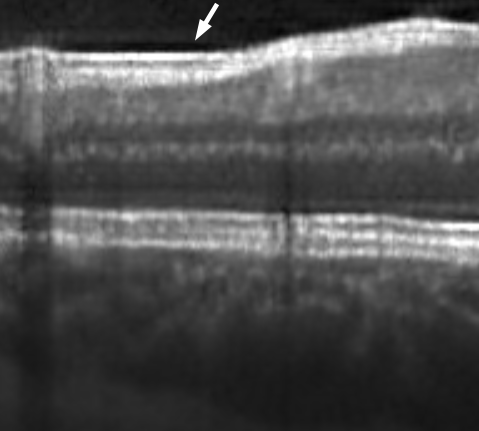} \\
    reconstruction}
    \caption{Hallucinations in reconstructed retinal OCT scan from supervisely trained CNN. (Left) Ground truth OCT scan. (Right) The white arrow denotes a hallucinated retinal layer that is anatomically incorrect. Hallucinations are the result of reconstructing an unseen noisy input using previously learned image statistics.}
    \label{fig:opener}
\end{figure}

To further increase the reliability in the denoised medical images, the reconstruction uncertainty has to be considered.
Bayesian autoencoders provide the mathematical framework to quantify a per-pixel reconstruction uncertainty \cite{Bishop2006,Kingma2013,Cheng2019}.
This allows the detection of hallucinations and other  artifacts, given that the uncertainty is well-calibrated; i.\,e.\ the uncertainty corresponds well with the reconstruction error \cite{Laves2020}.

In this work, we employ \emph{deep image prior} \cite{Lempitsky2018} to cope with hallucinations in medical image denoising and provide a Bayesian approach with Monte Carlo (MC) dropout \cite{Gal2016} that yields well-calibrated reconstruction uncertainty.
We present experimental results on denoising images from low-dose X-ray, ultrasound and OCT. 
Compared to previous work, our approach leads to better uncertainty estimates and is less prone to overfitting of the noisy image.
Our code is publicly available at \href{https://github.com/mlaves/uncertainty-deep-image-prior}{github.com/mlaves/uncertainty-deep-image-prior}.

\section{Related Work}

\paragraph{Image priors.}
Besides manually crafted priors such as 3D collaborative filtering \cite{Dabov2007}, convolutional denoising autoencoders have been used to implicitly learn an image prior from data \cite{Jain2009,Gondara2016}.
Lempitsky et al.\ have recently shown that the excellent performance of deep networks for inverse image tasks, such as denoising, is based not only on their ability to learn image priors from data, but also on the structure of a convolutional image generator itself \cite{Lempitsky2018}.
An image generator network $ \hat{\bm{x}} = \bm{f}_{\bm{\theta}}(\bm{z}) $ with randomly-initialized parameters $ \bm{\theta} $ is interpreted as parameterization of the image.
The parameters $ \bm{\theta} $ of the network are found by minimizing the pixel-wise squared error $ \Vert \tilde{\bm{x}} - \bm{f}_{\bm{\theta}}(\bm{z}) \Vert $ with stochastic gradient descent (SGD).
The input $ \bm{z} $ is sampled from a uniform distribution with additional perturbations by normally distributed noise in every iteration.
This is referred to as deep image prior (DIP).
They provided empirical evidence that the structure of a CNN alone is sufficient to capture enough image statistics to provide state-of-the-art performance in inverse imaging tasks.
During the process of SGD, low-frequency image features are reconstructed first, followed by higher frequencies, which makes human supervision necessary to retrieve the optimal denoised image.
Therefore, this approach heavily relies on early stopping in order to not overfit the noise.
However, a key advantage of deep image prior is the absence of hallucinations, since there is no prior learning.
A Bayesian approach could alleviate overfitting and additionally provide reconstruction uncertainty.

\paragraph{Bayesian deep learning.}
Bayesian neural networks allow estimation of predictive uncertainty \cite{Bishop2006} and we generally differentiate between aleatoric and epistemic uncertainty \cite{Kendall2017}.
Aleatoric uncertainty results from noise in the data (e.\,g.\ speckle noise in US or OCT).
It is derived from the conditional log-likelihood under the maximum likelihood estimation (MLE) or maximum posterior (MAP) framework and can be captured directly by a deep network (i.\,e.\ by subdividing the last layer of an image generator network).
Epistemic uncertainty is caused by uncertainty in the model parameters.
In deep learning, we usually perform MLE or MAP inference to find a single best estimate $ \hat{\bm{\theta}} $ for the network parameters.
This does not allow estimation of epistemic uncertainty and we therefore place distributions over the parameters.
In Bayesian inference, we want to consider all possible parameter configurations, weighted by their posterior.
Computing the posterior predictive distribution involves marginalization of the parameters $ \bm{\theta} $, which is intractable.
A common approximation of the posterior distribution is variational inference with Monte Carlo dropout \cite{Gal2016}.
It allows estimation of epistemic uncertainty by Monte Carlo sampling from the posterior of a network, that has been trained with dropout.

\paragraph{Bayesian deep image prior.}
Cheng et al.\ recently provided a Bayesian perspective on the deep image prior in the context of natural images, which is most related to our work \cite{Cheng2019}.
They interpret the convolutional network as spatial random process over the image coordinate space and use stochastic gradient Langevin dynamics (SGLD) as Bayesian approximation \cite{Welling2011} to sample from the posterior.
In SGLD, an MC sampler is derived from SGD by injecting Gaussian noise into the gradients after each backward pass.
The authors claim to have solved the overfitting issue with DIP and to be able to provide uncertainty estimates.
In the following, we will show that this is not the case for medical image denoising, even when using the code provided by the authors.
Further, the uncertainty estimates from SGLD do not reflect the predictive error with respect to the noise-free ground truth image.

\section{Methods}

\subsection{Aleatoric Uncertainty with Deep Image Prior}

We first revisit the concept of deep image prior for denoising and subsequently extend it to a Bayesian approach with Monte Carlo dropout to estimate both aleatoric and epistemic uncertainty.
Let $ \tilde{\bm{x}} $ be a noisy image, $ \bm{x} $ the true but generally unknown noise-free image and $ \bm{f}_{\bm{\theta}} $ an image generator network with parameter set $ \bm{\theta} $, that outputs the denoised image $ \hat{\bm{x}} $.
In deep image prior, the optimal parameter point estimate $ \hat{\bm{\theta}} $ is found by maximum likelihood estimation with gradient descent, which results in minimizing the squared error
\begin{equation}
    \hat{\bm{\theta}} = \argmin \Vert \tilde{\bm{x}} - \bm{f}_{\bm{\theta}}(\bm{z}) \Vert^{2}
\end{equation}
between the generated image $ \bm{f}_{\bm{\theta}} $ and the noisy image $ \tilde{\bm{x}} $.
The input $ \bm{z} \sim \mathcal{U}(0, 0.1) $ of the neural network has the same spatial dimensions as $ \tilde{\bm{x}} $ and is sampled from a uniform distribution.
To ensure that $ \hat{\bm{x}} $ has less noise, carefully chosen early stopping must be applied (see Sect.\,\ref{sec:results}).

To quantify aleatoric uncertainty, we assume that the image signal $ \tilde{\bm{x}} $ is sampled from a spatial random process and that each pixel $ i $ follows a Gaussian distribution $ \mathcal{N}(\tilde{x}_{i}; \hat{x}_{i}, \hat{\sigma}^{2}_{i}) $ with mean $ \hat{x}_{i} $ and variance $ \hat{\sigma}^{2}_{i} $.
We split the last layer such that the network outputs these values for each pixel
\begin{equation}
    \bm{f}_{\bm{\theta}} = \left[ \hat{\bm{x}}, \hat{\bm{\sigma}}^{2} \right] ~ .
\end{equation}
Now, MLE is performed by minimizing the full negative log-likelihood, which leads to the following optimization criterion \cite{Laves2020,Kendall2017}
\begin{equation}
    \mathcal{L}(\bm{\theta}) = \frac{1}{N}\sum_{i=1}^{N} \hat{\sigma}_{i}^{-2} \big\Vert \tilde{x}_{i} - \hat{x}_{i} \big\Vert^{2} + \log \hat{\sigma}_{i}^{2} ~ ,
    \label{eq:nll}
\end{equation}
where $ N $ is the number of pixels per image.
In this case, $ \hat{\bm{\sigma}}^{2} $ captures the pixel-wise aleatoric uncertainty and is jointly estimated with $ \hat{\bm{x}} $ by finding $ \bm{\theta} $ that minimizes Eq.\,(\ref{eq:nll}) with SGD.
For numerical stability, Eq.\,(\ref{eq:nll}) is implemented such that the network directly outputs $ -\log \hat{\bm{\sigma}}^{2} $.

\subsection{Epistemic Uncertainty with Bayesian Deep Image Prior}

Next, we move towards a Bayesian view to additionally quantify the epistemic uncertainty.
The image generator $ \bm{f}_{\bm{\theta}} $ is extended into a Bayesian neural network under the variational inference framework with MC dropout \cite{Gal2016}.
A prior distribution $ p(\bm{\theta}) \sim \mathcal{N}(\bm{0}, \lambda^{-1} \bm{I}) $ is placed over the parameters and the network $ \bm{f}_{\tilde{\bm{\theta}}} $ is trained with dropout by minimizing Eq.\,(\ref{eq:nll}) with added weight decay.
For inference, $ T $ stochastic forward passes with applied dropout are performed to sample from the approximate Bayesian posterior $ \tilde{\bm{\theta}} \sim q(\bm{\theta}) $.
This allows us to approximate the posterior predictive distribution
\begin{equation}
    p(\hat{\bm{x}} \vert \tilde{\bm{x}}) = \int p(\hat{\bm{x}} \vert \bm{\theta}, \tilde{\bm{x}}) p(\bm{\theta} \vert \tilde{\bm{x}}) \, \mathrm{d}\bm{\theta} ~ ,
\end{equation}
which is wider than the distribution from MLE or MAP, as it accounts for uncertainty in $ \bm{\theta} $.
We use Monte Carlo integration to estimate the predictive mean
\begin{equation}
    \hat{\bm{x}} = \frac{1}{T} \sum_{t=1}^{T} \hat{\bm{x}}_{t}
\end{equation}
and predictive variance \cite{Laves2020,Kendall2017}
\begin{equation}
    \hat{\bm{\sigma}}^{2} =  \underbrace{ \frac{1}{T} \sum_{t=1}^{T} \left( \hat{\bm{x}}_{t} - \frac{1}{T} \sum_{t=1}^{T} \hat{\bm{x}}_{t} \right)^{2}}_{\mathrm{epistemic}} + \underbrace{ \frac{1}{T} \sum_{t=1}^{T} \hat{\bm{\sigma}}^{2}_{t} }_{\mathrm{aleatoric}}
    \label{eq:pred_variance}
\end{equation}
with $ \bm{f}_{\tilde{\bm{\theta}}_t} = [ \bm{\hat{x}}_{t}, \bm{\hat{\sigma}}^{2}_{t} ] $.
In this work, we use $ T=25 $ MC samples with dropout probability of $ p = 0.3 $.
The resulting $ \hat{\bm{x}} $ is used as estimation of the noise-free image and $ \hat{\bm{\sigma}}^{2} $ is used as uncertainty map.
We use the mean over the pixel coordinates as scalar uncertainty value $ U $.

\subsection{Calibration of Uncertainty}

Following recent literature, we define predictive uncertainty to be well-calibrated if it correlates linearly with the predictive error \cite{Guo2017,Levi2019,Laves2020}.
More formally, miscalibration is quantified with
\begin{equation}
    \mathbb{E}_{\hat{\sigma}^{2}} \left[ \big\vert \big( \Vert \tilde{\bm{x}} - \hat{\bm{x}} \Vert^{2} \, \big\vert \, \hat{\sigma}^{2} = \sigma^{2} \big) - \sigma^{2} \big\vert \right] \quad \Forall \left\{ \sigma^{2} \in \mathbb{R} \, \vert \, \sigma^{2} \geq 0 \right\} ~ .
    \label{eq:miscalibration}
\end{equation}
That is, if all pixels in a batch were estimated with uncertainty of $ 0.2 $, we expect the predictive error (MSE) to also equal $ 0.2 $.
To approximate Eq.\,(\ref{eq:miscalibration}) on an image with finite pixels, we use the uncertainty calibration error (UCE) metric presented in \cite{Laves2020}, which involves binning the uncertainty values and computing a weighted average of absolute differences between MSE and uncertainty per bin.

\section{Experiments}

\begin{figure}[t]
    \centering
    \parbox[t]{1.95cm}{\centering
        \includegraphics[width=1.9cm]{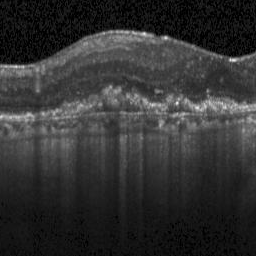} \\
        $ \bm{x}_{\mathrm{OCT}} $
    }
    \parbox[t]{1.95cm}{\centering
        \includegraphics[width=1.9cm]{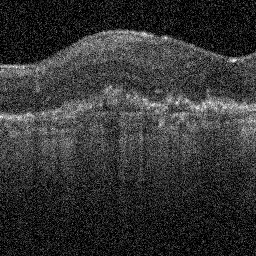} \\
        $ \tilde{\bm{x}}_{\mathrm{OCT}} $
    } 
    \parbox[t]{1.95cm}{\centering
        \includegraphics[width=1.9cm]{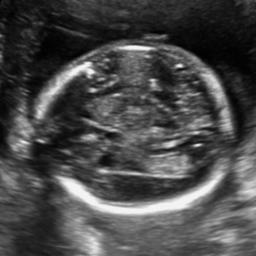} \\
        $ \bm{x}_{\mathrm{US}} $
    }
    \parbox[t]{1.95cm}{\centering
        \includegraphics[width=1.9cm]{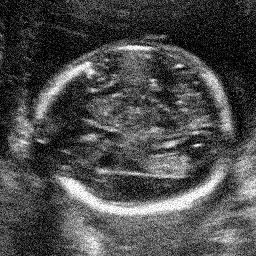} \\
        $ \tilde{\bm{x}}_{\mathrm{US}} $
    }
    \parbox[t]{1.95cm}{\centering
        \includegraphics[width=1.9cm]{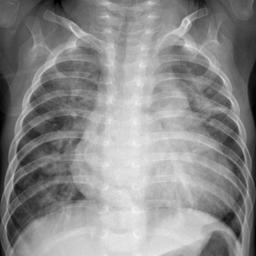} \\
        $ \bm{x}_{\mathrm{xray}} $
    }
    \parbox[t]{1.95cm}{\centering
        \includegraphics[width=1.9cm]{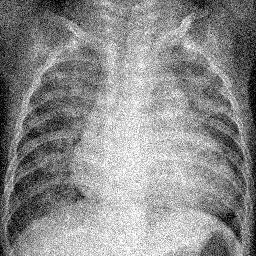} \\
        $ \tilde{\bm{x}}_{\mathrm{xray}} $
    }
    \caption{Images used to evaluate the denoising performance. The task is to reconstruct a noise-free image from $ \tilde{\bm{x}} $ without having access to $ \bm{x} $. OCT and US images are characterized by speckle noise which can be simulated by additive Gaussian noise. Low-dose X-ray shows uneven photon density that can be simulated by Poisson noise.}
    \label{fig:denoising}
\end{figure}

We refer to the presented Bayesian approach to deep image prior with Monte Carlo dropout as MCDIP and evaluate its denoising performance and the calibration of uncertainty on three different medical imaging modalities (see Fig.\,\ref{fig:denoising}).
The first test image $ \bm{x}_{\mathrm{OCT}} $ shows an OCT scan of a retina affected by choroidal neovascularization.
Next, $ \bm{x}_{\mathrm{US}} $ shows an ultrasound of a fetal head for gestational age estimation.
The third test image $ \bm{x}_{\mathrm{xray}} $ shows a chest x-ray for pneumonia assessment.
All test images are arbitrarily sampled from public data sets \cite{Kermany2018,Heuvel2018} and have a resolution of $ 512 \times 512 $ pixel.

Images from optical coherence tomography and ultrasound are prone to speckle noise due to interference phenomena \cite{Michailovich2006}.
Speckle noise can obscure small anatomical details and reduce image contrast.
It is worth mentioning that speckle patterns also contain information about the microstructure of the tissue.
However, this information is not perceptible to a human observer, therefore the denoising of such images is desirable.
Noise in low-dose X-ray originates from an uneven photon density and can be modeled with Poisson noise \cite{Zabic2013,Lee2018}.
In this work, we approximate the Poisson noise with Gaussian noise since $ \mathsf{Poisson}(\lambda) $ approaches a Normal distribution as $ \lambda \rightarrow \infty $ (see Appendix\,\ref{app:poisson}).
We first create a low-noise image $ \bm{x} $ by smoothing and downsampling the original image from public data sets using the \texttt{ANTIALIAS} filter from the Python Imaging Library (\texttt{PIL}) to $ 256 \times 256 $ pixel.
Downsampling involves averaging over highly correlated neighboring pixels affected by uncorrelated noise.
This decreases the observation noise by sacrificing image resolution (see Appendix\,\ref{app:downsampling}).
The downsampled image acts as ground truth to which we compute the peak signal-to-noise ratio (PSNR) and the structural similarity (SSIM) of the denoised image $ \hat{\bm{x}} $.
Further, we compute the UCE and provide calibration diagrams (MSE vs.\ uncertainty) to show the (mis-)calibration of the uncertainty estimates.

We compare the results from MCDIP to standard DIP and to DIP with SGLD from Cheng et al. \cite{Cheng2019}.
SGLD posterior inference is performed by averaging over $ T $ posterior samples $ \hat{\bm{x}} = \frac{1}{T} \sum_{t=1}^{T} \hat{\bm{x}}_{t} $ after a ``burn in'' phase.
The posterior variance is used as an estimator of the epistemic uncertainty $ \frac{1}{T} \sum_{t=1}^{T} \left( \hat{\bm{x}} - \hat{\bm{x}}_{t} \right)^{2} $.
Cheng et al. claim that their approach does not require early stopping and yields better denoising performance.
Additionally, we train the SGLD approach with the loss function from Eq.\,(\ref{eq:pred_variance}) to consider aleatoric uncertainty and denote this with SGLD+NLL.
We implement SGLD using the Adam optimizer, which works better in practice and is more related to preconditioned SGLD \cite{Li2016}.

\section{Results}
\label{sec:results}

The results are presented threefold: We show (1) possible overfitting in Fig.\,\ref{fig:psnrs} by plotting the PSNR between the reconstruction $ \hat{\bm{x}} $ and the ground truth image $ \bm{x} $; (2) denoising performance by providing the denoised images in Fig.\,\ref{fig:denoising_results} and PSNR in Tab.\,\ref{tab:psnr} after convergence (i.\,e.\ after 50k optimizer steps); and (3) goodness of uncertainty in Fig.\,\ref{fig:calib} by providing calibration diagrams and uncertainty maps.

Our experiments confirm what is already known: The non-Bayesian DIP quickly overfits the noisy image.
The narrow peaks in PSNR values during optimization show that manually performed early stopping is essential to obtain a reconstructed image with less noise (see Fig.\,\ref{fig:psnrs}).
The PSNR between $ \hat{\bm{x}} $ and the ground truth $ \bm{x} $ approaches the value of the PSNR between the noisy image $ \tilde{\bm{x}} $ and the ground truth, thus reconstructing the noise as well.
However, the SGLD approach shows almost identical overfitting behavior in our experiments.
This is in contrast to what is stated by Chen et al., even when using the original implementation of SGLD provided by the authors \cite{Cheng2019}.
SGLD+NLL additionally considers aleatoric uncertainty and converges to a higher PSNR level.
This indicates that SGLD+NLL does not overfit the noisy image completely.
MCDIP on the other hand does not show a sharp peak in Fig.\,\ref{fig:psnrs} and safely converges to its highest PSNR value.
This requires no manual early stopping to obtain a denoised image.
The reconstructed X-ray images after convergence in Fig.\,\ref{fig:denoising_results} underline this: MCDIP does not reconstruct the noise.
The PSNR values in Tab.\,\ref{tab:psnr} confirm these observations.
Although it was not the intention of this work to reach highest-possible PSNR values, MCDIP even outperforms the other methods with early-stopping applied (see Appendix \ref{app:tables}).

\begin{figure}[t]
    \centering
    \includegraphics[scale=0.48]{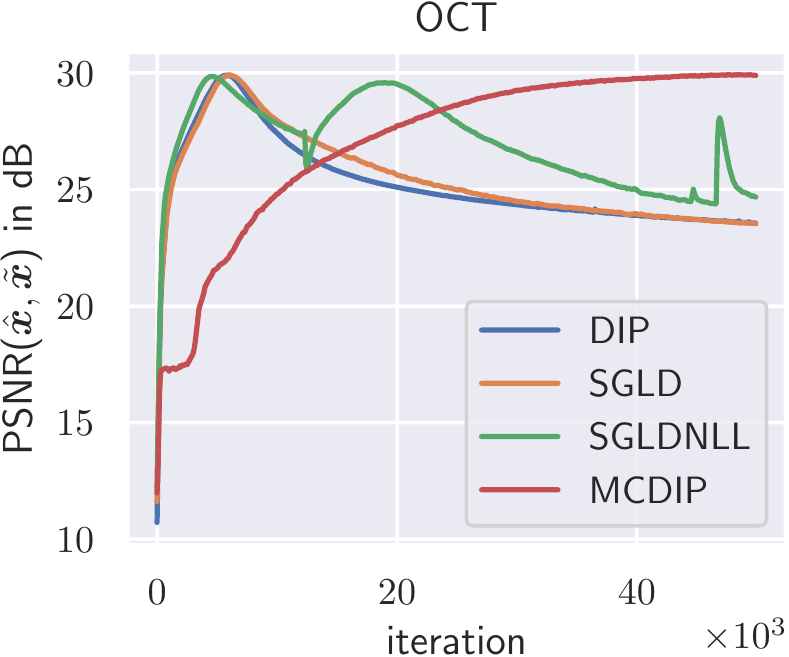}
    \includegraphics[scale=0.48]{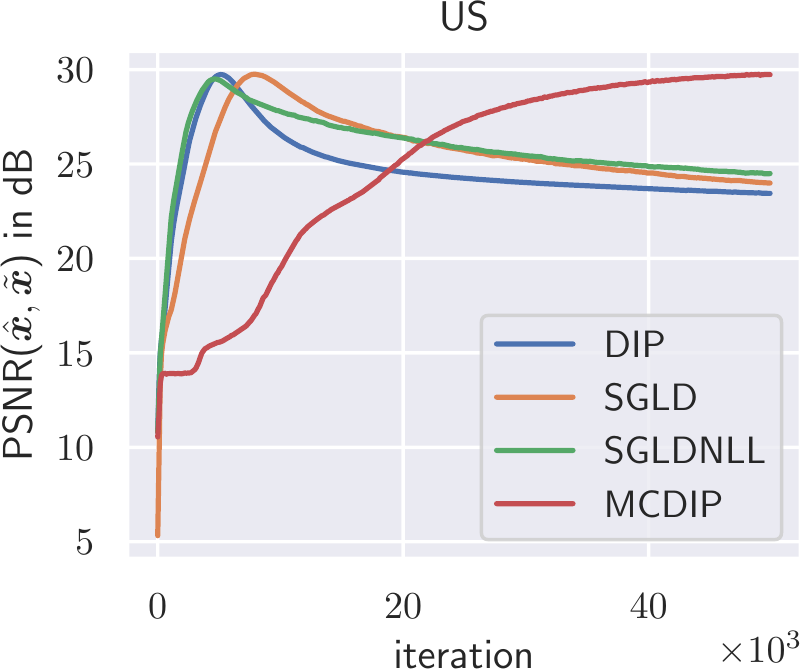}
    \includegraphics[scale=0.48]{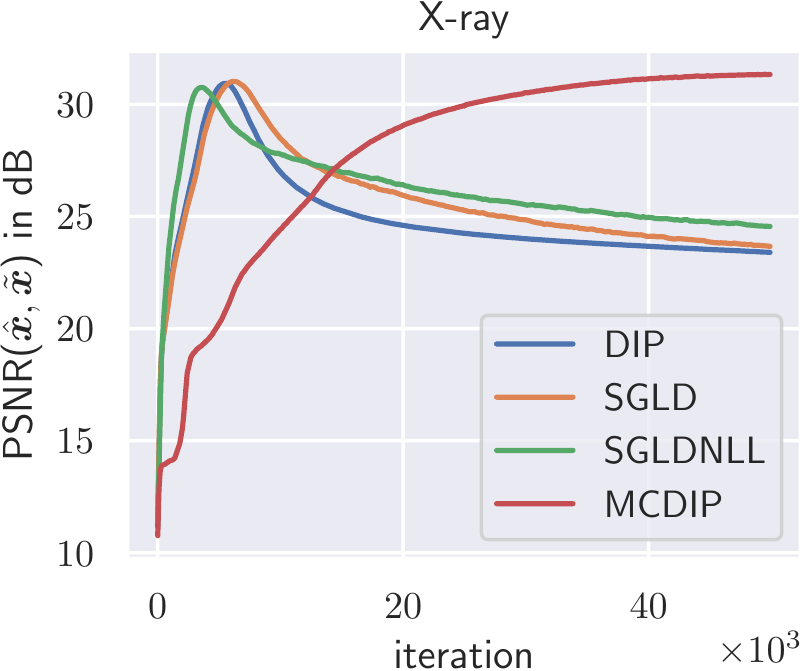}
    \caption{Peak signal-to-noise ratio between denoised image $ \hat{\bm{x}} $ and ground truth $ \bm{x} $ vs.\ number of optimizer iterations. DIP and SGLD(+NLL) quickly overfit the noisy image. MCDIP converges to its highest PSNR value and does not overfit $ \tilde{\bm{x}} $. The plots show means from 3 runs with different random initialization.}
    \label{fig:psnrs}
\end{figure}

\begin{figure}[t]
    \centering
    \parbox[t]{2.25cm}{\centering
        \includegraphics[width=2.2cm]{BACTERIA-1351146-0006.jpg} \\
        ground truth
    }
    \parbox[t]{2.25cm}{\centering
        \includegraphics[width=2.2cm]{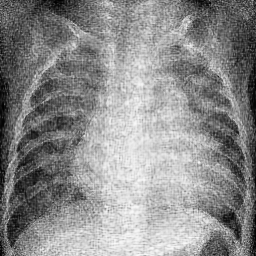} \\
        DIP
    }
    \parbox[t]{2.25cm}{\centering
        \includegraphics[width=2.2cm]{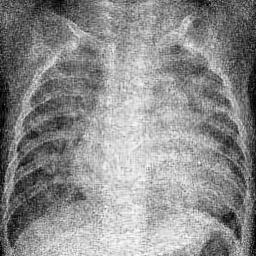} \\
        SGLD
    }
    \parbox[t]{2.25cm}{\centering
        \includegraphics[width=2.2cm]{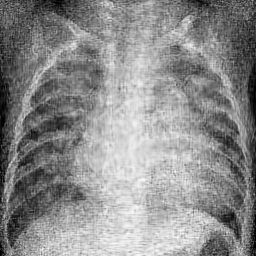} \\
        SGLD+NLL
    }
    \parbox[t]{2.25cm}{\centering
        \includegraphics[width=2.2cm]{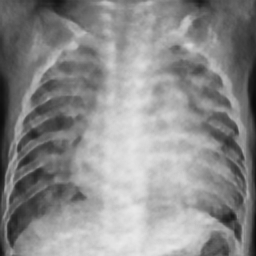} \\
        MCDIP
    }
    \caption{Denoised X-ray images after convergence. Only MCDIP does not show overfitted noise. Additional reconstructions can be found in Appendix \ref{app:figures}.}
    \label{fig:denoising_results}
\end{figure}

\begingroup
\setlength{\tabcolsep}{6pt}
\begin{table}[t]
    \centering
    \caption{PSNR values after convergence (at least 50k iterations). Note that our goal was not to reach highest possible PSNR, but to show overfitting in convergence.}
    \label{tab:psnr}
    \begin{tabular}{rcccc}
        \toprule
        PSNR   & DIP & SGLD & SGLD+NLL & MCDIP \\
        \cmidrule{2-5}
        OCT    & $ 23.64 \pm 0.19 $ & $ 23.58 \pm 0.12 $ & $ 24.82 \pm 0.12 $ & $ \mathbf{29.88} \pm 0.03 $ \\
        US     & $ 23.55 \pm 0.11 $ & $ 23.81 \pm 0.15 $ & $ 24.55 \pm 0.08 $ & $ \mathbf{29.67} \pm 0.07 $ \\
        X-ray  & $ 23.28 \pm 0.08 $ & $ 23.50 \pm 0.12 $ & $ 24.60 \pm 0.04 $ & $ \mathbf{31.19} \pm 0.10 $ \\
        \bottomrule
    \end{tabular} ~ 
\end{table}
\endgroup

\begin{figure}
    \centering
    \includegraphics[height=2.6cm]{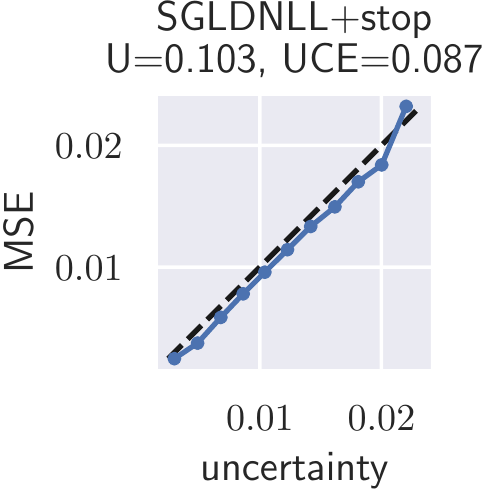}
    \includegraphics[height=2.6cm]{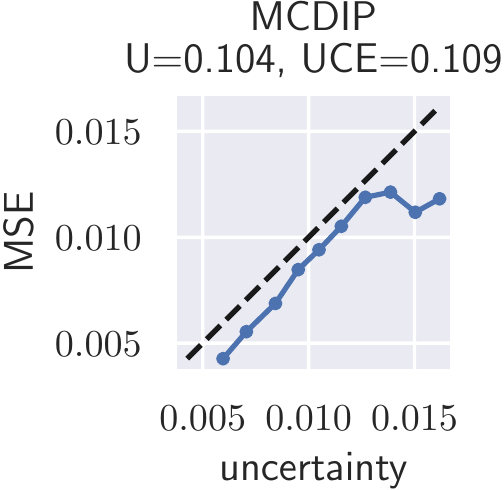}
    \includegraphics[height=2.6cm]{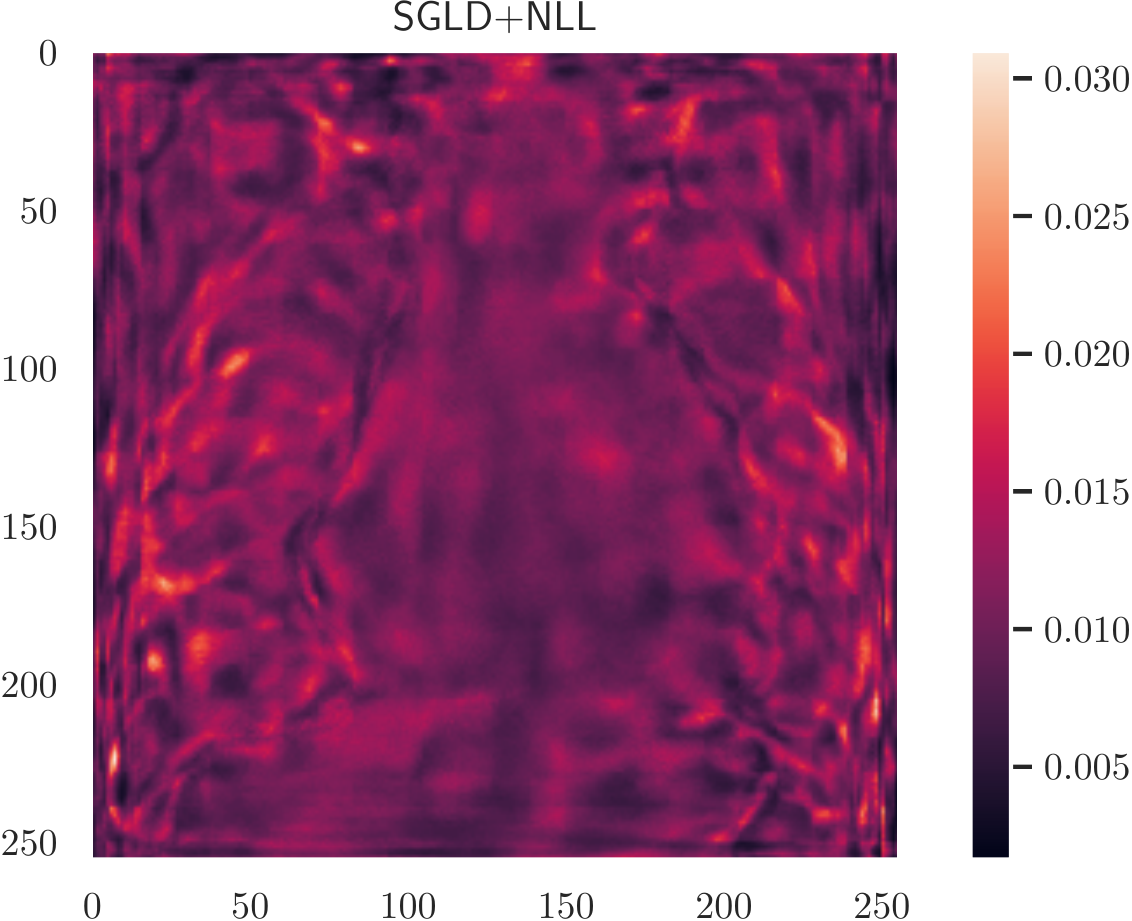}
    \includegraphics[height=2.6cm]{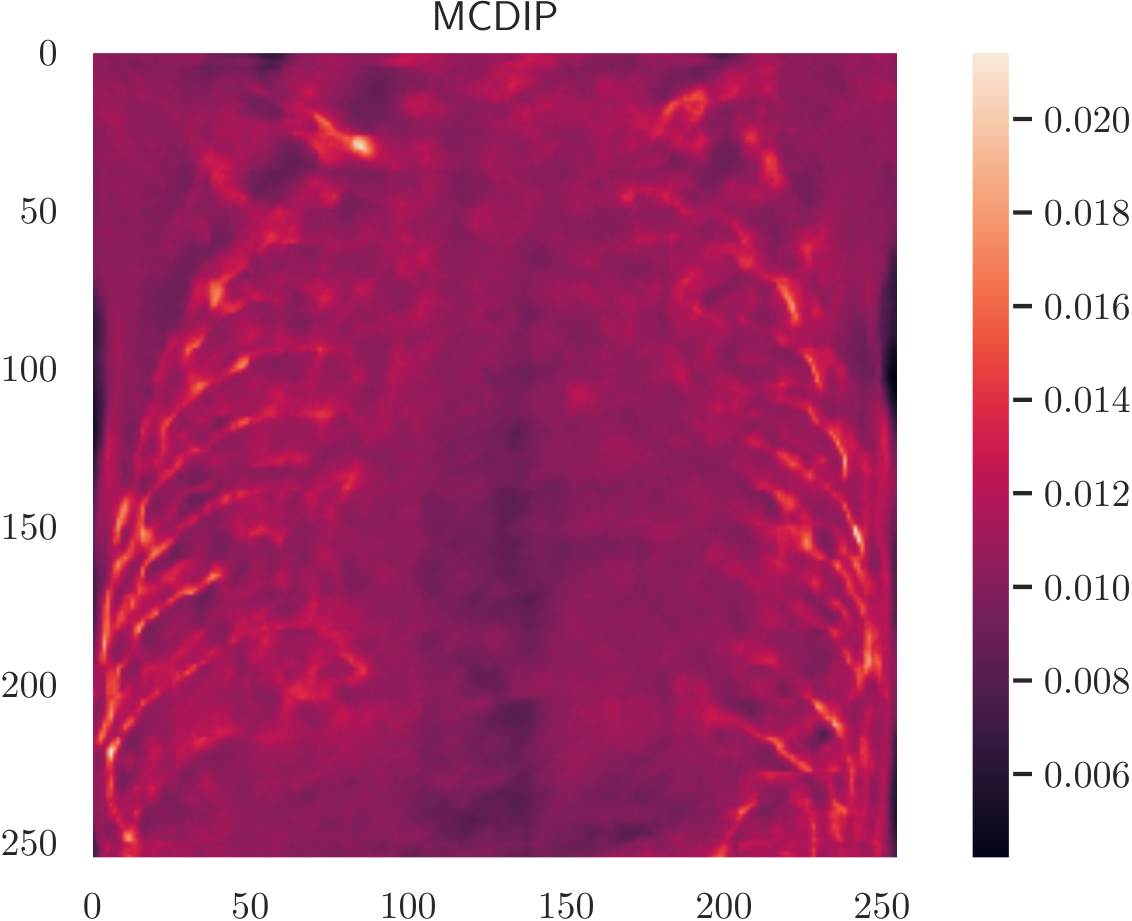} \\
    \caption{Calibration diagrams and uncertainty maps for SGLD+NLL with early stopping and MCDIP after convergence on the X-ray image (best viewed with digital zoom). (Left) The calibration diagrams show MSE vs.\ uncertainty and provide mean uncertainty (U) and UCE values. (Right) Uncertainty maps show per-pixel uncertainty.}
    \label{fig:calib}
\end{figure}

The calibration diagrams and corresponding UCE values in Fig.\,\ref{fig:calib} suggest that SGLD+NLL is better calibrated than MCDIP.
However, due to overfitting the noisy image without early stopping, the MSE from SGLD+NLL concentrates around $ 0.0 $, which results in low UCE values.
On the US and OCT image, the uncertainty from SGLD+NLL collapses to a single bin in the calibration diagram and does not allow to reason about the validness of the reconstructed image (see Fig.\,\ref{fig:calib2} in Appendix~\ref{app:figures}).
The uncertainty map from MCDIP shows high uncertainty at edges in the image and the mean uncertainty value (denoted by U) is close to the noise level in all three test images.

\section{Discussion \& Conclusion}

In this paper, we provided a new Bayesian approach to the deep image prior.
We used variational inference with Monte Carlo dropout and the full negative log-likelihood to both quantify epistemic and aleatoric uncertainty.
The presented approach is applied to medical image denoising of three different modalities and provides state-of-the-art performance in denoising with deep image prior.
Our Bayesian treatment does not need carefully applied early stopping and yields well-calibrated uncertainty.
We observe the estimated mean uncertainty value to be close to the noise level of the images.

The question remains why Bayesian deep image prior with SGLD does not work as well as expected and is outperformed by MC dropout.
First, SGLD as described by Welling et al.\ requires a strong decay of the step size to ensure convergence to a mode of the posterior \cite{Welling2011}.
Cheng et al.\ did not implement this and we followed their approach \cite{Cheng2019}.
After implementing the described step size decay, SGLD did not overfit the noisy image (see Appendix\,\ref{app:sgldlr}).
However, this requires a carefully chosen step size decay which is equivalent to early stopping.


The deep image prior framework is especially interesting in medical imaging as it does not require supervised training and thus does not suffer from hallucinations and other artifacts.
The presented approach can further be applied to deformable registration or other inverse image tasks in the medical domain.

\bibliographystyle{splncs04}
\bibliography{bibliography}

\begin{thebibliography}{10}
\providecommand{\url}[1]{\texttt{#1}}
\providecommand{\urlprefix}{URL }
\providecommand{\doi}[1]{https://doi.org/#1}

\bibitem{Agostinelli2013}
Agostinelli, F., Anderson, M.R., Lee, H.: Adaptive multi-column deep neural
  networks with application to robust image denoising. In: Advances in Neural
  Information Processing Systems. pp. 1493--1501 (2013)

\bibitem{Bishop2006}
Bishop, C.M.: Pattern Recognition and Machine Learning. Springer (2006)

\bibitem{Chang2000}
Chang, S.G., Yu, B., Vetterli, M.: {Adaptive wavelet thresholding for image
  denoising and compression}. {IEEE Transactions on Image Processing}
  \textbf{9}(9),  1532--1546 (2000). \doi{10.1109/83.862633}

\bibitem{Cheng2019}
Cheng, Z., Gadelha, M., Maji, S., Sheldon, D.: A bayesian perspective on the
  deep image prior. In: IEEE/CVF Conference on Computer Vision and Pattern
  Recognition. pp. 5443--5451 (2019)

\bibitem{Dabov2007}
Dabov, K., Foi, A., Katkovnik, V., Egiazarian, K.: Image denoising by sparse
  3-d transform-domain collaborative filtering. Transactions on Image
  Processing  \textbf{16}(8),  2080--2095 (2007). \doi{10.1109/TIP.2007.901238}

\bibitem{Gal2016}
Gal, Y., Ghahramani, Z.: Dropout as a bayesian approximation: Representing
  model uncertainty in deep learning. In: ICML. pp. 1050--1059 (2016)

\bibitem{Gondara2016}
Gondara, L.: Medical image denoising using convolutional denoising
  autoencoders. In: International Conference on Data Mining Workshops. pp.
  241--246 (2016). \doi{10.1109/ICDMW.2016.0041}

\bibitem{Guo2017}
Guo, C., Pleiss, G., Sun, Y., Weinberger, K.Q.: On calibration of modern neural
  networks. In: ICML. pp. 1321--1330 (2017)

\bibitem{Heuvel2018}
van~den Heuvel, T.L., de~Bruijn, D., de~Korte, C.L., Ginneken, B.v.: Automated
  measurement of fetal head circumference using 2d ultrasound images. PloS one
  \textbf{13}(8),  e0200412 (2018). \doi{10.1371/journal.pone.0200412}, uS
  dataset source

\bibitem{Hogg2018}
Hogg, R.V., McKean, J., Craig, A.T.: Introduction to Mathematical Statistics.
  Pearson, 8 edn. (2018)

\bibitem{Jain2009}
Jain, V., Seung, S.: Natural image denoising with convolutional networks. In:
  Advances in Neural Information Processing Systems. pp. 769--776 (2009)

\bibitem{Kendall2017}
Kendall, A., Gal, Y.: What uncertainties do we need in bayesian deep learning
  for computer vision? In: NeurIPS. pp. 5574--5584 (2017)

\bibitem{Kermany2018}
Kermany, D.S., Goldbaum, M., Cai, W., Valentim, C.C., Liang, H., Baxter, S.L.,
  McKeown, A., Yang, G., Wu, X., Yan, F., Dong, J., Prasadha, M.K., Pei, J.,
  Ting, M.Y., Zhu, J., Li, C., Hewett, S., Dong, J., Ziyar, I., Shi, A., Zhang,
  R., Zheng, L., Hou, R., Shi, W., Fu, X., Duan, Y., Huu, V.A., Wen, C., Zhang,
  E.D., Zhang, C.L., Li, O., Wang, X., Singer, M.A., Sun, X., Xu, J., Tafreshi,
  A., Lewis, M.A., Xia, H., Zhang, K.: Identifying medical diagnoses and
  treatable diseases by image-based deep learning. Cell  \textbf{172}(5),
  1122--1131 (2018). \doi{10.1016/j.cell.2018.02.010}

\bibitem{Kingma2013}
Kingma, D.P., Welling, M.: Auto-encoding variational bayes. In: ICLR (2014)

\bibitem{Laves2020}
Laves, M.H., Ihler, S., Fast, J.F., Kahrs, L.A., Ortmaier, T.: Well-calibrated
  regression uncertainty in medical imaging with deep learning. In: Medical
  Imaging with Deep Learning (2020)

\bibitem{Laves2019ECBO}
Laves, M.H., Ihler, S., Kahrs, L.A., Ortmaier, T.: Semantic denoising
  autoencoders for retinal optical coherence tomography. In: SPIE/OSA European
  Conference on Biomedical Optics. vol. 11078, pp. 86--89 (2019).
  \doi{10.1117/12.2526936}

\bibitem{Lee2018}
Lee, S., Lee, M.S., Kang, M.G.: Poisson--gaussian noise analysis and estimation
  for low-dose x-ray images in the nsct domain. Sensors  \textbf{18}(4), ~1019
  (2018)

\bibitem{Lempitsky2018}
Lempitsky, V., Vedaldi, A., Ulyanov, D.: {Deep Image Prior}. In: IEEE/CVF
  Conference on Computer Vision and Pattern Recognition. pp. 9446--9454 (2018).
  \doi{10.1109/CVPR.2018.00984}

\bibitem{Levi2019}
{Levi}, D., {Gispan}, L., {Giladi}, N., {Fetaya}, E.: Evaluating and
  calibrating uncertainty prediction in regression tasks. In: arXiv (2019),
  arXiv:1905.11659

\bibitem{Li2016}
Li, C., Chen, C., Carlson, D., Carin, L.: Preconditioned stochastic gradient
  langevin dynamics for deep neural networks. In: Proceedings of the Thirtieth
  AAAI Conference on Artificial Intelligence. pp. 1788--1794 (2016)

\bibitem{Michailovich2006}
Michailovich, O.V., Tannenbaum, A.: Despeckling of medical ultrasound images.
  Transactions on Ultrasonics, Ferroelectrics, and Frequency Control
  \textbf{53}(1),  64--78 (2006). \doi{10.1109/TUFFC.2006.1588392}

\bibitem{Rabbani2009}
Rabbani, H., Nezafat, R., Gazor, S.: Wavelet-domain medical image denoising
  using bivariate laplacian mixture model. Transactions on Biomedical
  Engineering  \textbf{56}(12),  2826--2837 (2009).
  \doi{10.1109/TBME.2009.2028876}

\bibitem{Salinas2007}
Salinas, H.M., Fernandez, D.C.: {Comparison of PDE-Based Nonlinear Diffusion
  Approaches for Image Enhancement and Denoising in Optical Coherence
  Tomography}. {IEEE Transactions on Medical Imaging}  \textbf{26}(6),
  761--771 (2007). \doi{10.1109/TMI.2006.887375}

\bibitem{Sotiras2013}
Sotiras, A., Davatzikos, C., Paragios, N.: Deformable medical image
  registration: A survey. IEEE Transactions on Medical Imaging  \textbf{32}(7),
   1153--1190 (2013). \doi{10.1109/TMI.2013.2265603}

\bibitem{Wang2014}
Wang, N., Tao, D., Gao, X., Li, X., Li, J.: A comprehensive survey to face
  hallucination. International Journal of Computer Vision  \textbf{106}(1),
  9--30 (2014)

\bibitem{Welling2011}
Welling, M., Teh, Y.W.: Bayesian learning via stochastic gradient langevin
  dynamics. In: ICML. pp. 681--688 (2011)

\bibitem{Zabic2013}
{\v{Z}}abi{\'c}, S., Wang, Q., Morton, T., Brown, K.M.: A low dose simulation
  tool for ct systems with energy integrating detectors. Medical physics
  \textbf{40}(3),  031102 (2013). \doi{10.1118/1.4789628}

\bibitem{Zhang2017}
{Zhang}, K., {Zuo}, W., {Chen}, Y., {Meng}, D., {Zhang}, L.: Beyond a gaussian
  denoiser: Residual learning of deep cnn for image denoising. IEEE
  Transactions on Image Processing  \textbf{26}(7),  3142--3155 (2017).
  \doi{10.1109/TIP.2017.2662206}

\end{thebibliography}

\newpage

\appendix

\section{Appendix}

\subsection{Additional Figures}
\label{app:figures}
\enlargethispage{1cm}

\begin{figure}[!h]
    \centering
    \parbox[t]{2.25cm}{\centering
        \includegraphics[width=2.2cm]{CNV-9997680-30.png} \\
        \vspace{1mm}
        \includegraphics[width=2.2cm]{081_HC.jpg} \\
        \vspace{1mm}
        \includegraphics[width=2.2cm]{BACTERIA-1351146-0006.jpg} \\
        ground truth
    }
    \parbox[t]{2.25cm}{\centering
        \includegraphics[width=2.2cm]{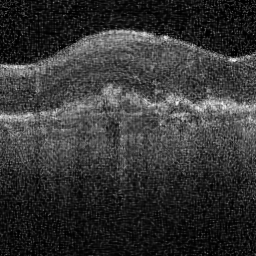} \\
        \vspace{1mm}
        \includegraphics[width=2.2cm]{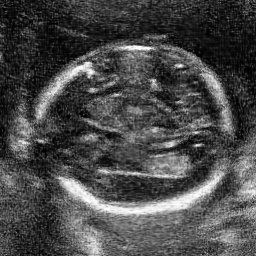} \\
        \vspace{1mm}
        \includegraphics[width=2.2cm]{recon_xray_dip.png} \\
        DIP
    }
    \parbox[t]{2.25cm}{\centering
        \includegraphics[width=2.2cm]{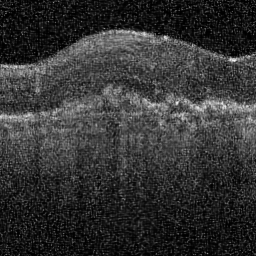} \\
        \vspace{1mm}
        \includegraphics[width=2.2cm]{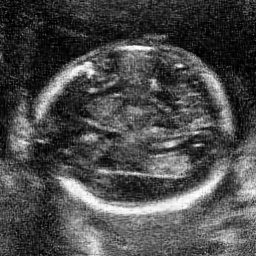} \\
        \vspace{1mm}
        \includegraphics[width=2.2cm]{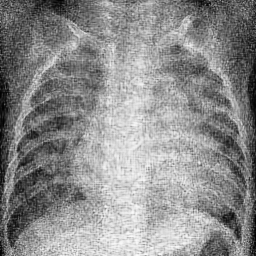} \\
        SGLD
    }
    \parbox[t]{2.25cm}{\centering
        \includegraphics[width=2.2cm]{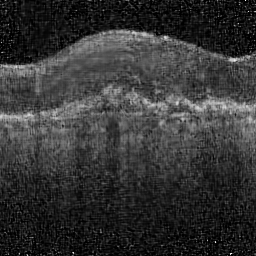} \\
        \vspace{1mm}
        \includegraphics[width=2.2cm]{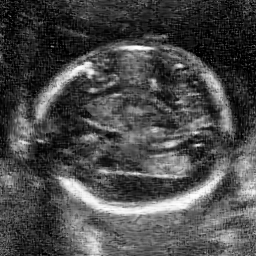} \\
        \vspace{1mm}
        \includegraphics[width=2.2cm]{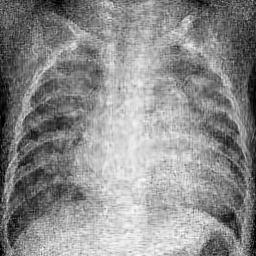} \\
        SGLD+NLL
    }
    \parbox[t]{2.25cm}{\centering
        \includegraphics[width=2.2cm]{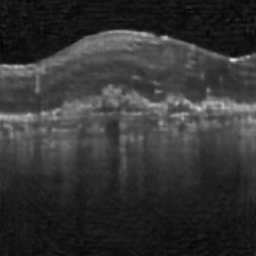} \\
        \vspace{1mm}
        \includegraphics[width=2.2cm]{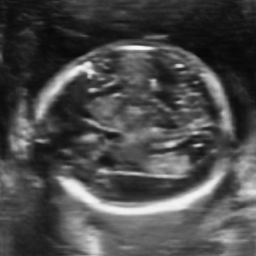} \\
        \vspace{1mm}
        \includegraphics[width=2.2cm]{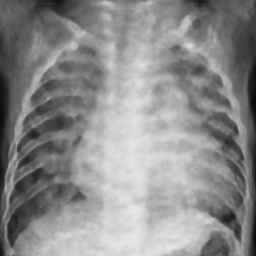} \\
        MCDIP
    }
    \caption{Denoised images after convergence.}
    \label{fig:denoising_results_2}
\end{figure}
\begin{figure}[!h]
    \centering
    \parbox[t]{2.25cm}{\centering
        \includegraphics[width=2.2cm]{CNV-9997680-30.png} \\
        \vspace{1mm}
        \includegraphics[width=2.2cm]{081_HC.jpg} \\
        \vspace{1mm}
        \includegraphics[width=2.2cm]{BACTERIA-1351146-0006.jpg} \\
        ground truth
    }
    \parbox[t]{2.25cm}{\centering
        \includegraphics[width=2.2cm]{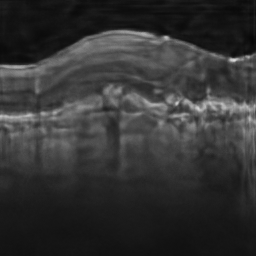} \\
        \vspace{1mm}
        \includegraphics[width=2.2cm]{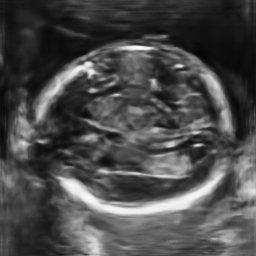} \\
        \vspace{1mm}
        \includegraphics[width=2.2cm]{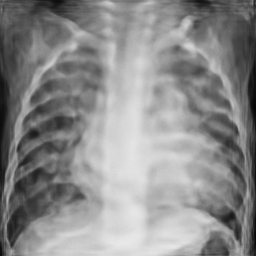} \\
        DIP
    }
    \parbox[t]{2.25cm}{\centering
        \includegraphics[width=2.2cm]{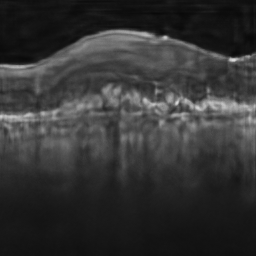} \\
        \vspace{1mm}
        \includegraphics[width=2.2cm]{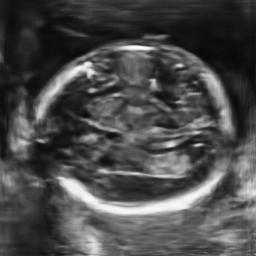} \\
        \vspace{1mm}
        \includegraphics[width=2.2cm]{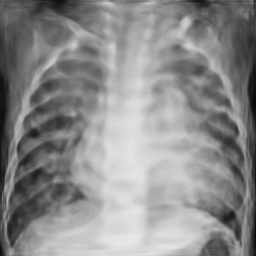} \\
        SGLD
    }
    \parbox[t]{2.25cm}{\centering
        \includegraphics[width=2.2cm]{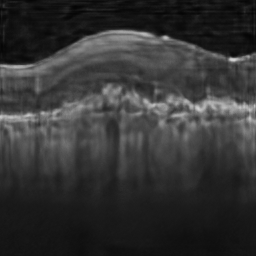} \\
        \vspace{1mm}
        \includegraphics[width=2.2cm]{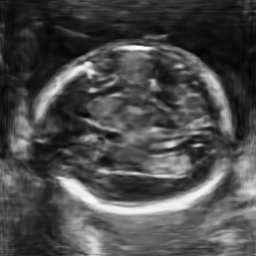} \\
        \vspace{1mm}
        \includegraphics[width=2.2cm]{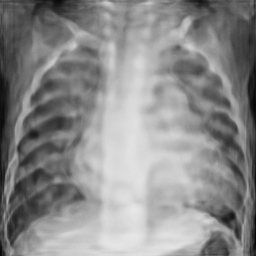} \\
        SGLD+NLL
    }
    \parbox[t]{2.25cm}{\centering
        \includegraphics[width=2.2cm]{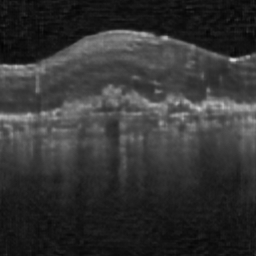} \\
        \vspace{1mm}
        \includegraphics[width=2.2cm]{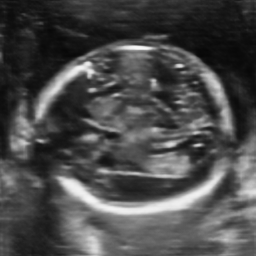} \\
        \vspace{1mm}
        \includegraphics[width=2.2cm]{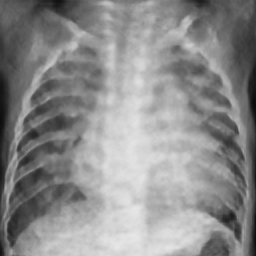} \\
        MCDIP
    }
    \caption{Denoised images with early-stopping applied.}
    \label{fig:denoising_results_early_stopping}
\end{figure}
\pagebreak

\begin{figure}
    \centering
    \includegraphics[scale=0.48]{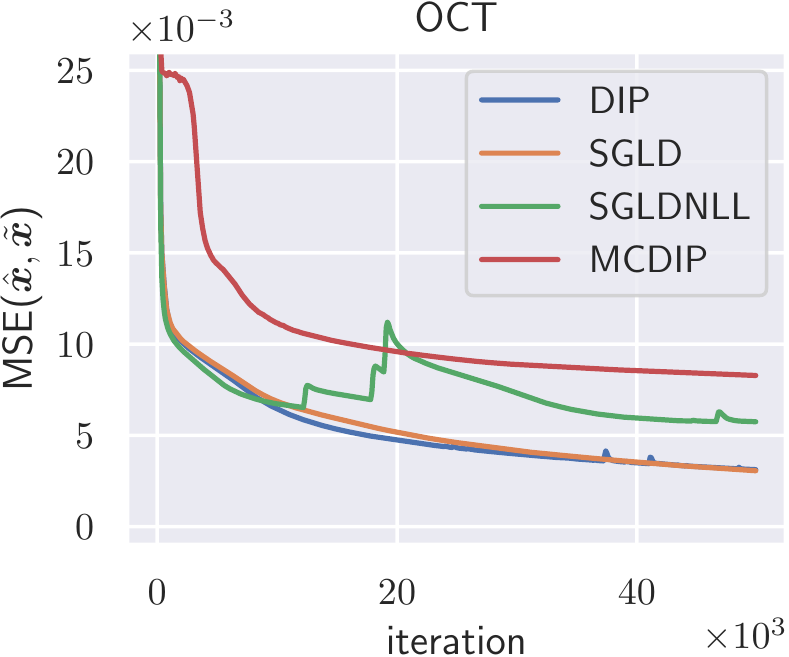} \includegraphics[scale=0.48]{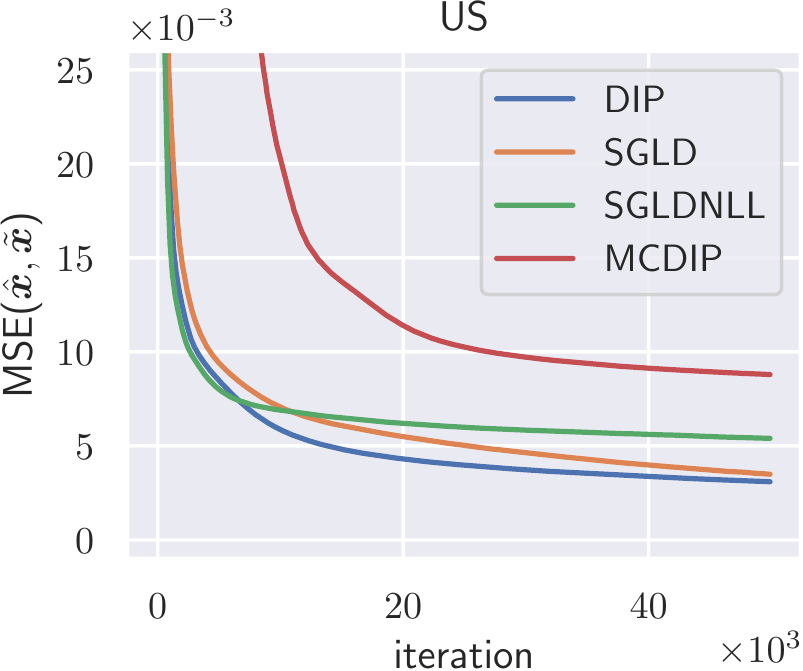} \includegraphics[scale=0.48]{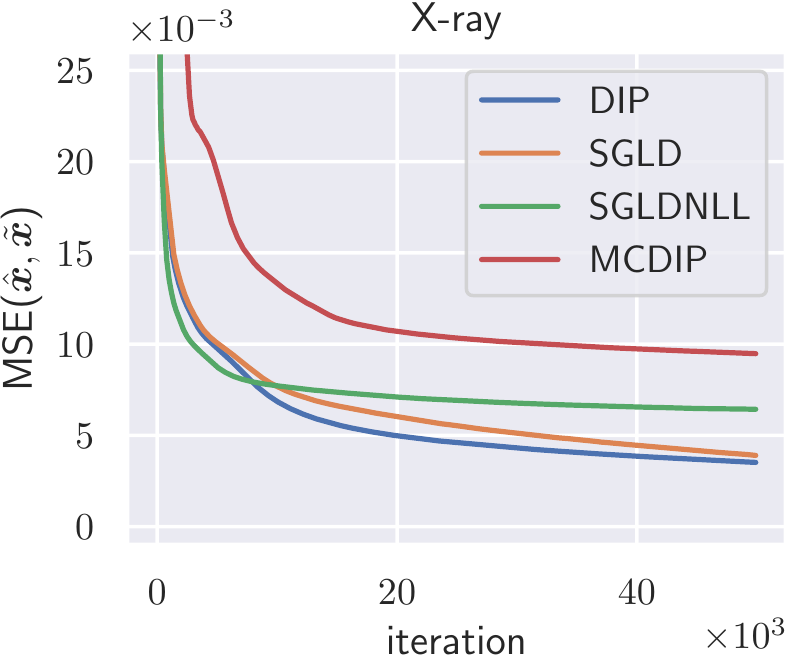} \\
    \vspace{2mm}
    \includegraphics[scale=0.48]{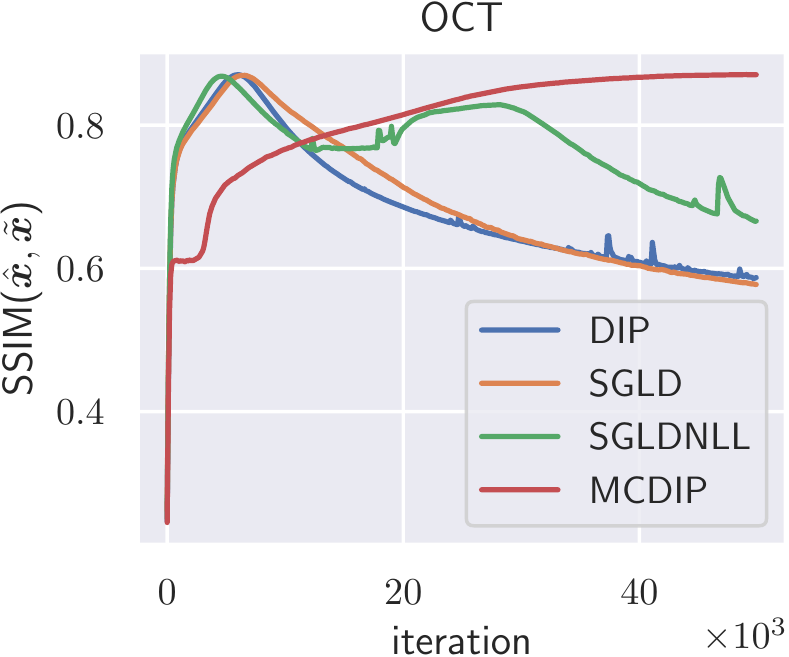}
    \includegraphics[scale=0.48]{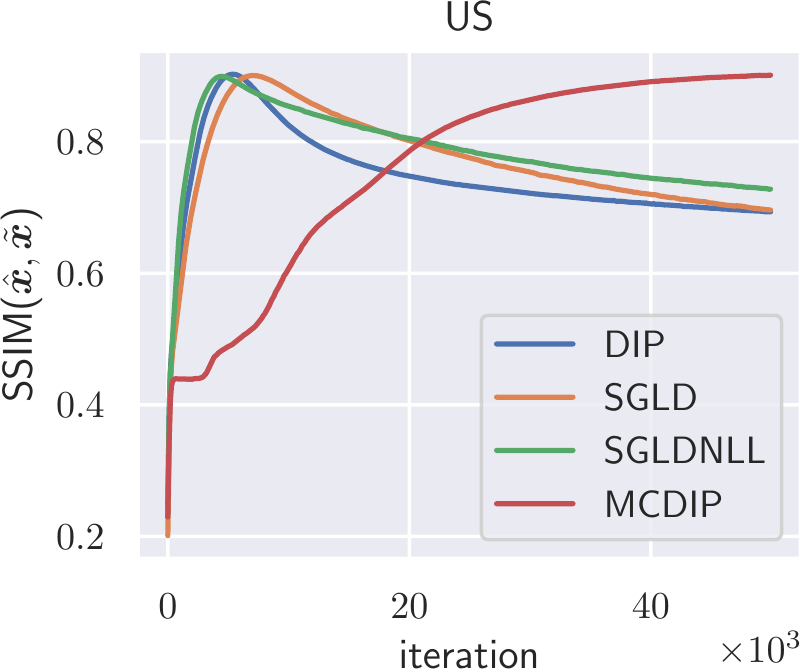}
    \includegraphics[scale=0.48]{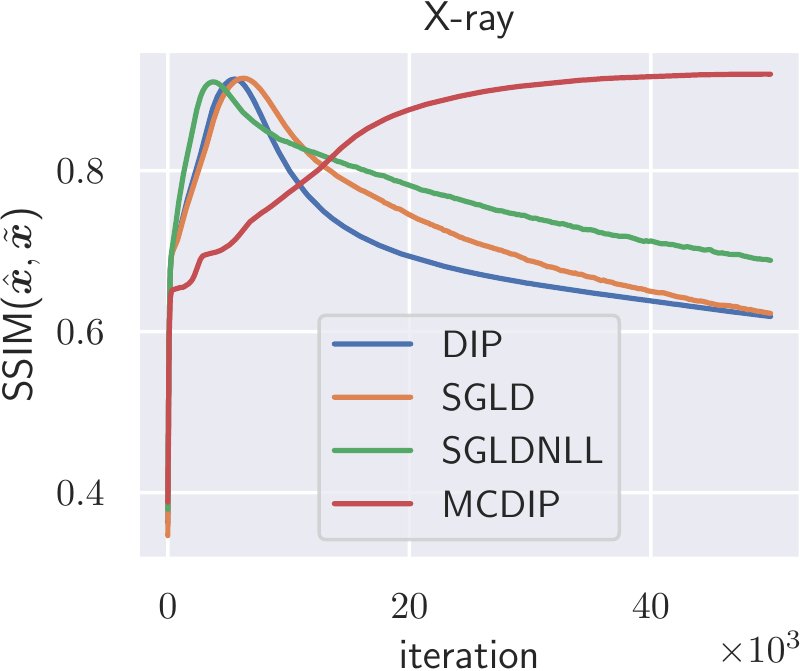}
   \caption{MSE (top row) between denoised $ \hat{\bm{x}} $ image and noisy image $ \tilde{\bm{x}} $ and SSIM (bottom row) between denoised $ \hat{\bm{x}} $ image and ground truth $ \bm{x} $ vs.\ iteration. Only MCDIP does not overfit the noisy image and converges with highest similarity to the ground truth. Despite the claim of the authors, SGLD suffers from overfitting and creates the need for carefully applied early stopping \cite{Cheng2019}. Note: We compared both our own implementation of SGLD and the original code provided by the authors. The plots show means from 3 runs with different random initialization.}
    \label{fig:mse_ssim}
\end{figure}

\begin{figure}
    \centering
    \includegraphics[height=2.3cm]{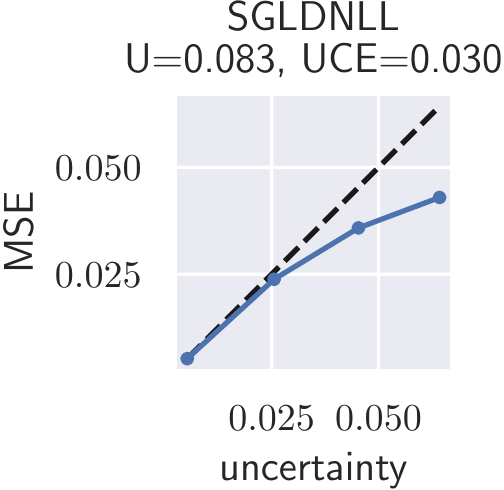}
    \includegraphics[height=2.3cm]{calib_mcdip_X-ray-crop.pdf}
    \includegraphics[height=2.3cm]{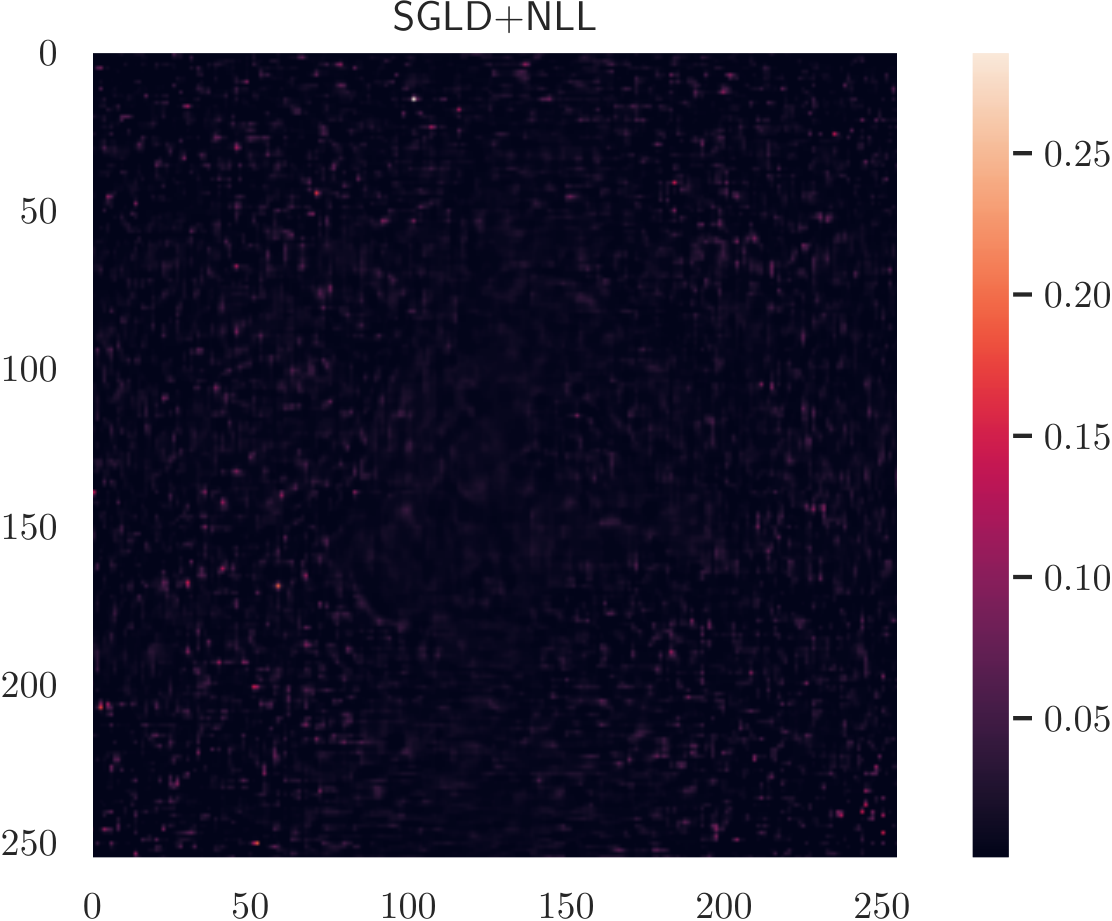}
    \includegraphics[height=2.3cm]{uncert_map_X-ray_mcdip-crop.pdf} \\
    \includegraphics[height=2.3cm]{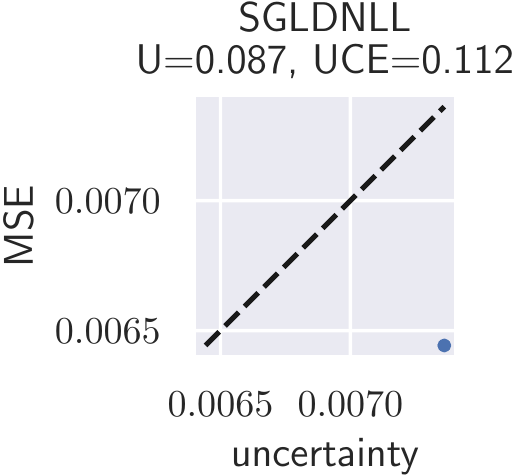}
    \includegraphics[height=2.3cm]{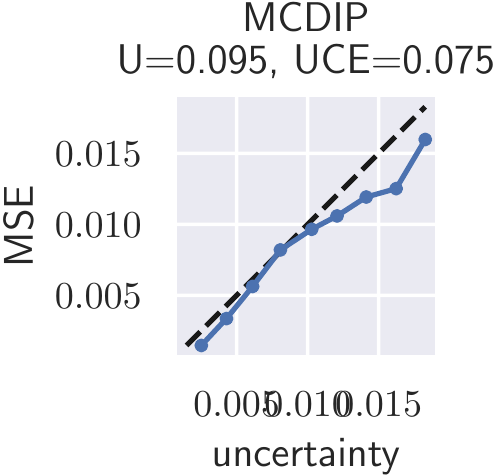}
    \includegraphics[height=2.3cm]{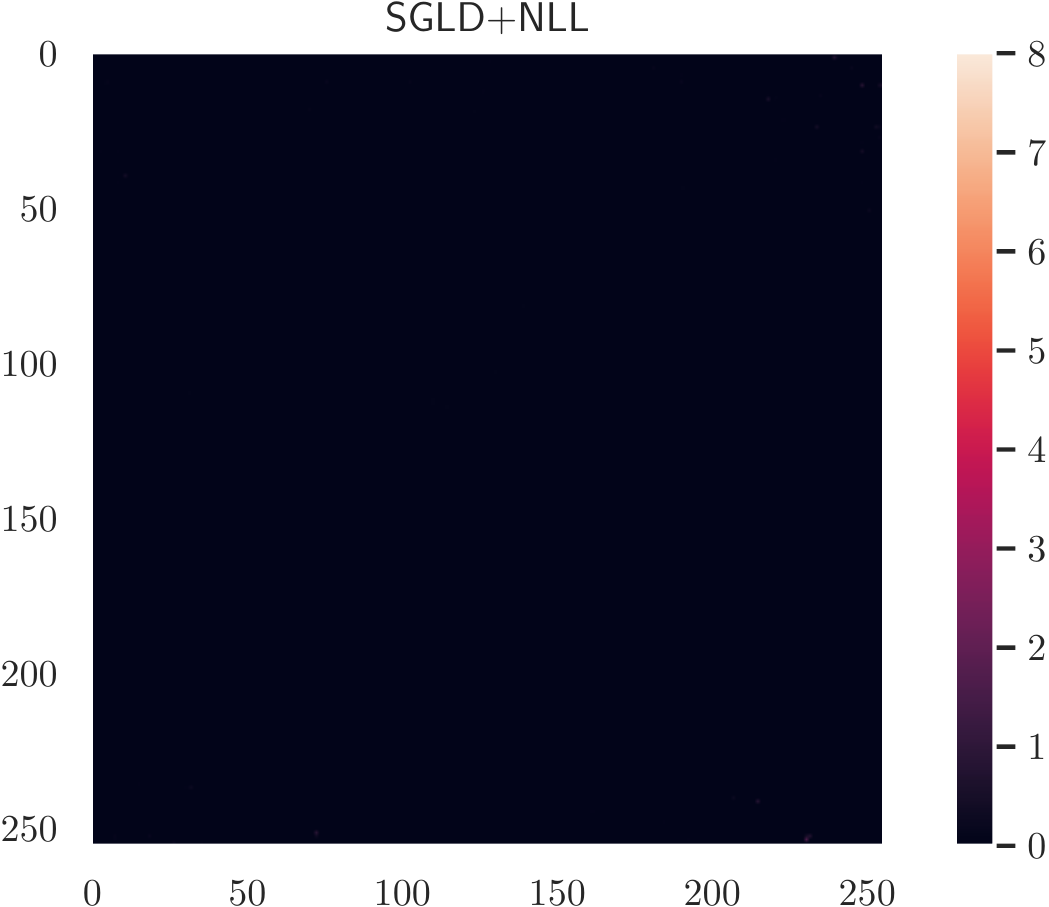}
    \includegraphics[height=2.3cm]{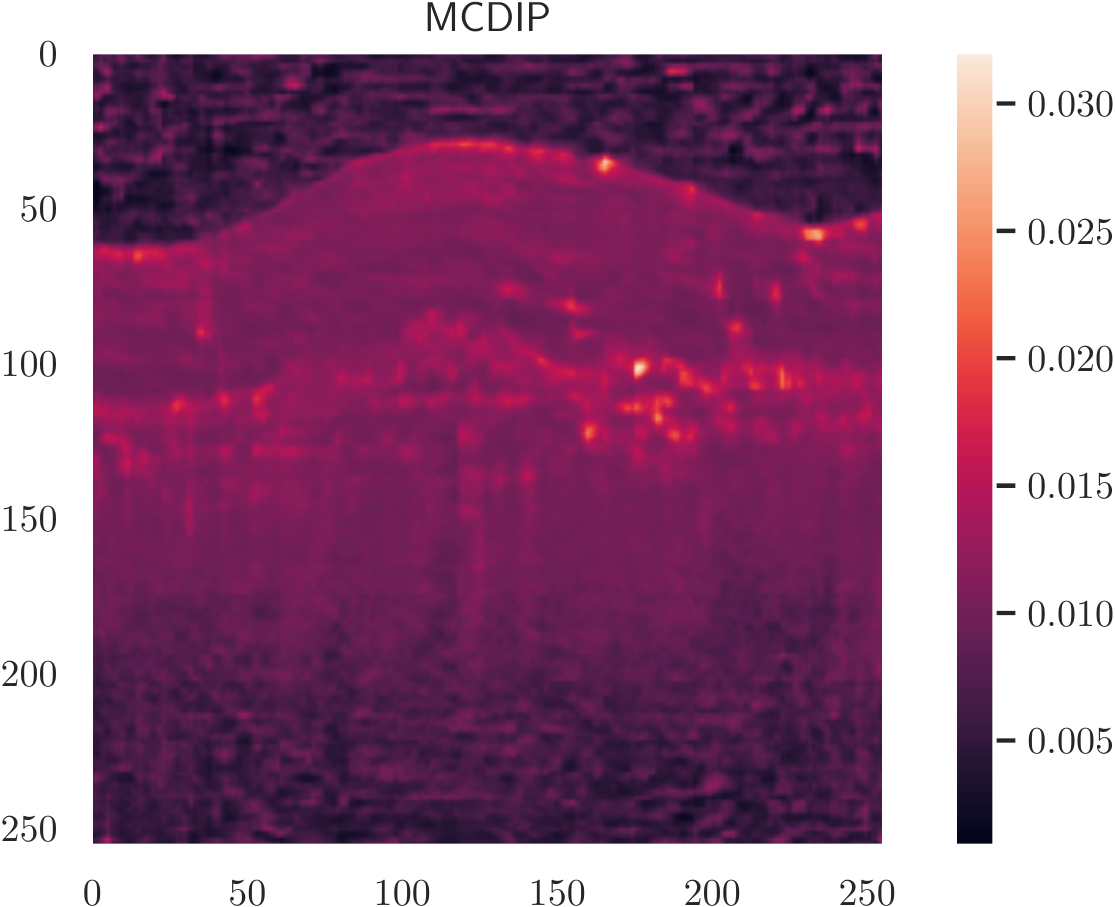} \\
    \includegraphics[height=2.3cm]{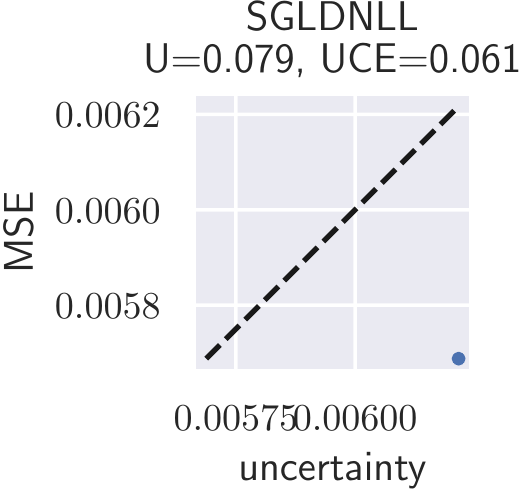}
    \includegraphics[height=2.3cm]{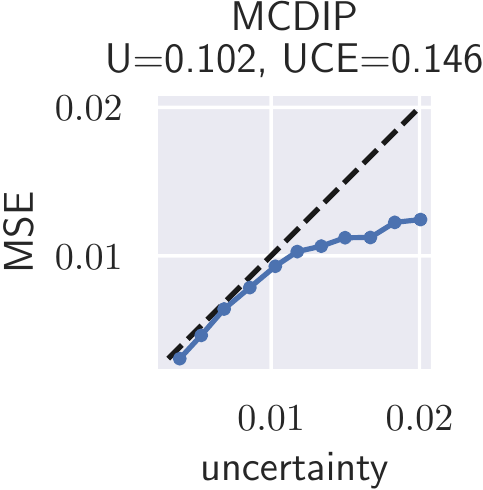}
    \includegraphics[height=2.3cm]{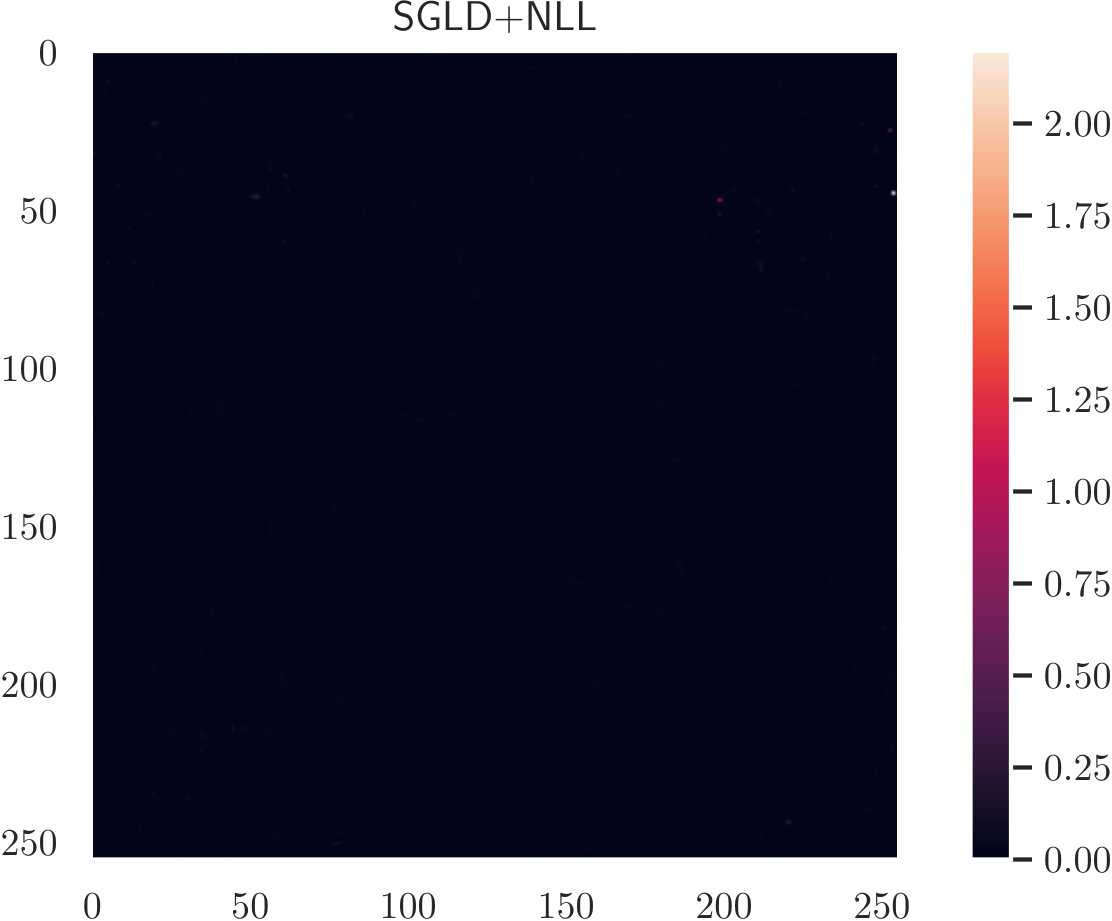}
    \includegraphics[height=2.3cm]{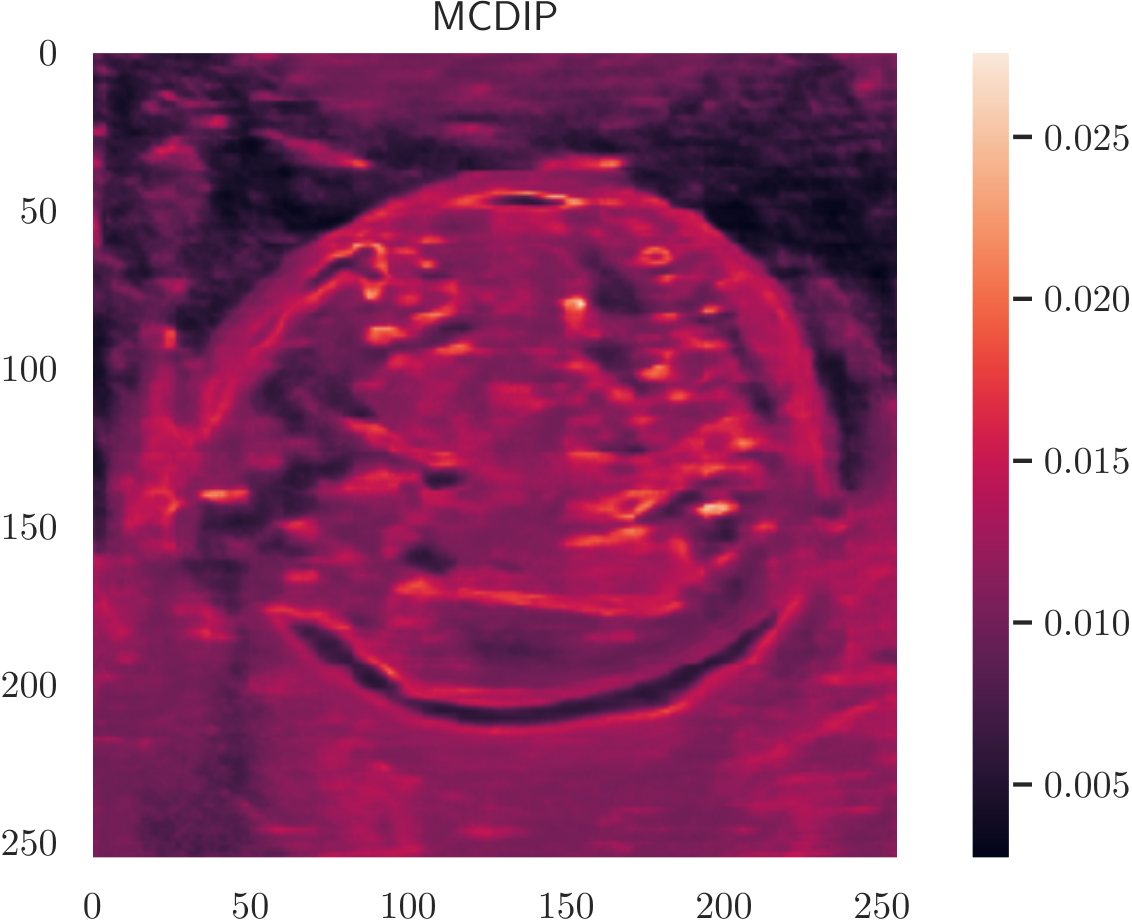}
    \caption{Calibration diagrams and uncertainty maps for SGLD+NLL and MCDIP after convergence (best viewed with digital zoom). (Left) The calibration diagrams show MSE vs.\ uncertainty and provide mean uncertainty (U) and UCE values. (Right) Uncertainty maps show per-pixel uncertainty. Due to overfitting, the MSE and uncertainty from SGLD+NLL concentrates around $0.0$.}
    \label{fig:calib2}
\end{figure}

\begin{figure}
    \centering
    \includegraphics[height=2.3cm]{calib_sgldnll_argmax_X-ray-crop.pdf}
    \includegraphics[height=2.3cm]{calib_mcdip_X-ray-crop.pdf}
    \includegraphics[height=2.3cm]{uncert_map_X-ray_sgldnll_earlystop-crop.pdf}
    \includegraphics[height=2.3cm]{uncert_map_X-ray_mcdip-crop.pdf} \\
    \includegraphics[height=2.3cm]{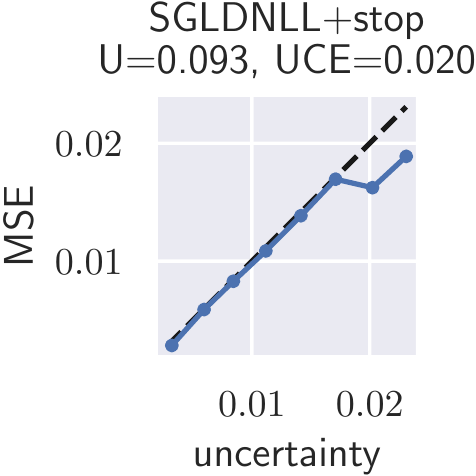}
    \includegraphics[height=2.3cm]{calib_mcdip_OCT-crop.pdf}
    \includegraphics[height=2.3cm]{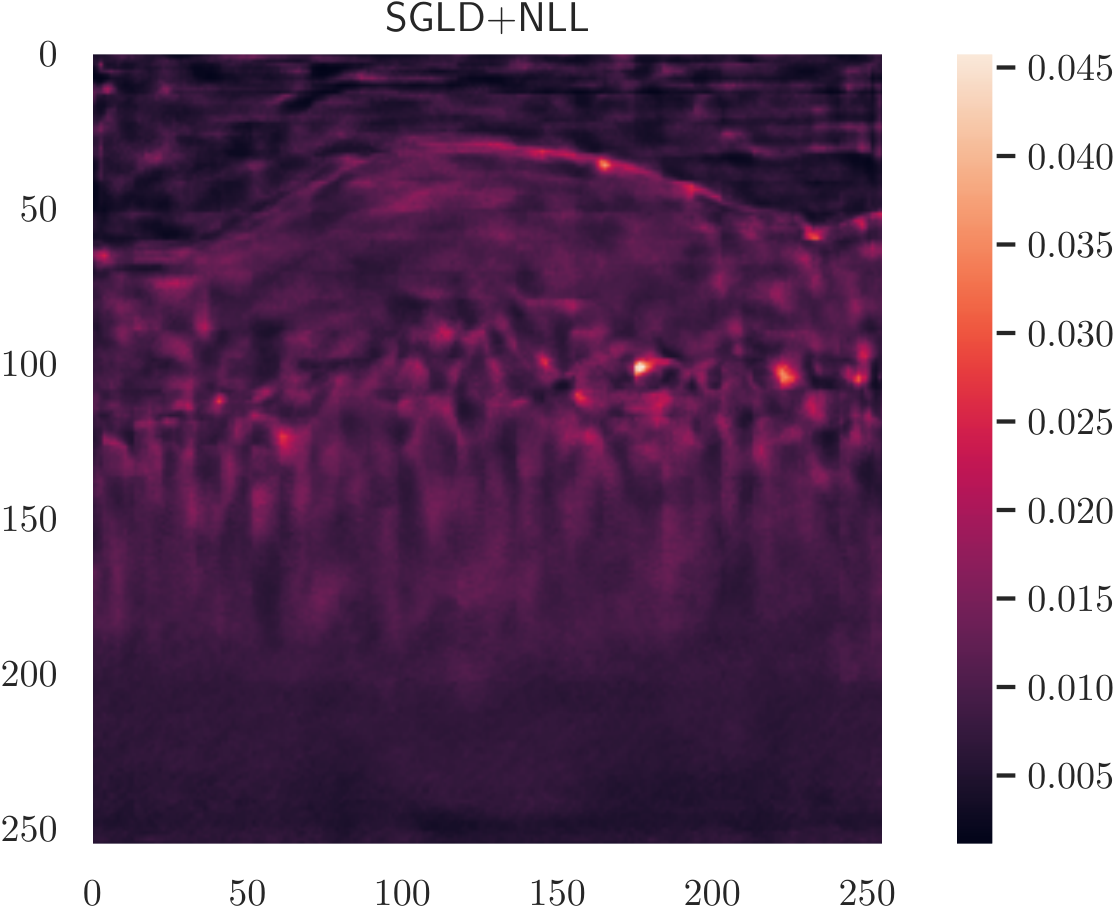}
    \includegraphics[height=2.3cm]{uncert_map_OCT_mcdip-crop.pdf} \\
    \includegraphics[height=2.3cm]{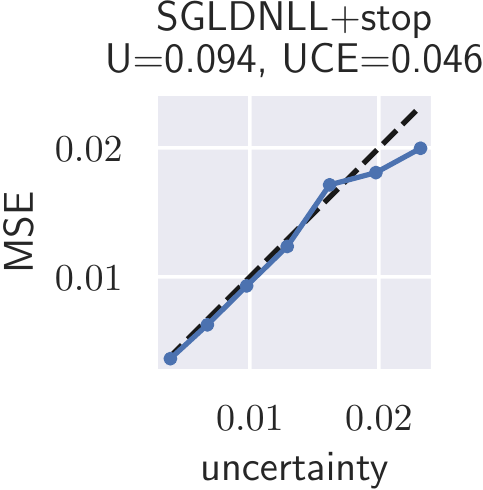}
    \includegraphics[height=2.3cm]{calib_mcdip_us-crop.pdf}
    \includegraphics[height=2.3cm]{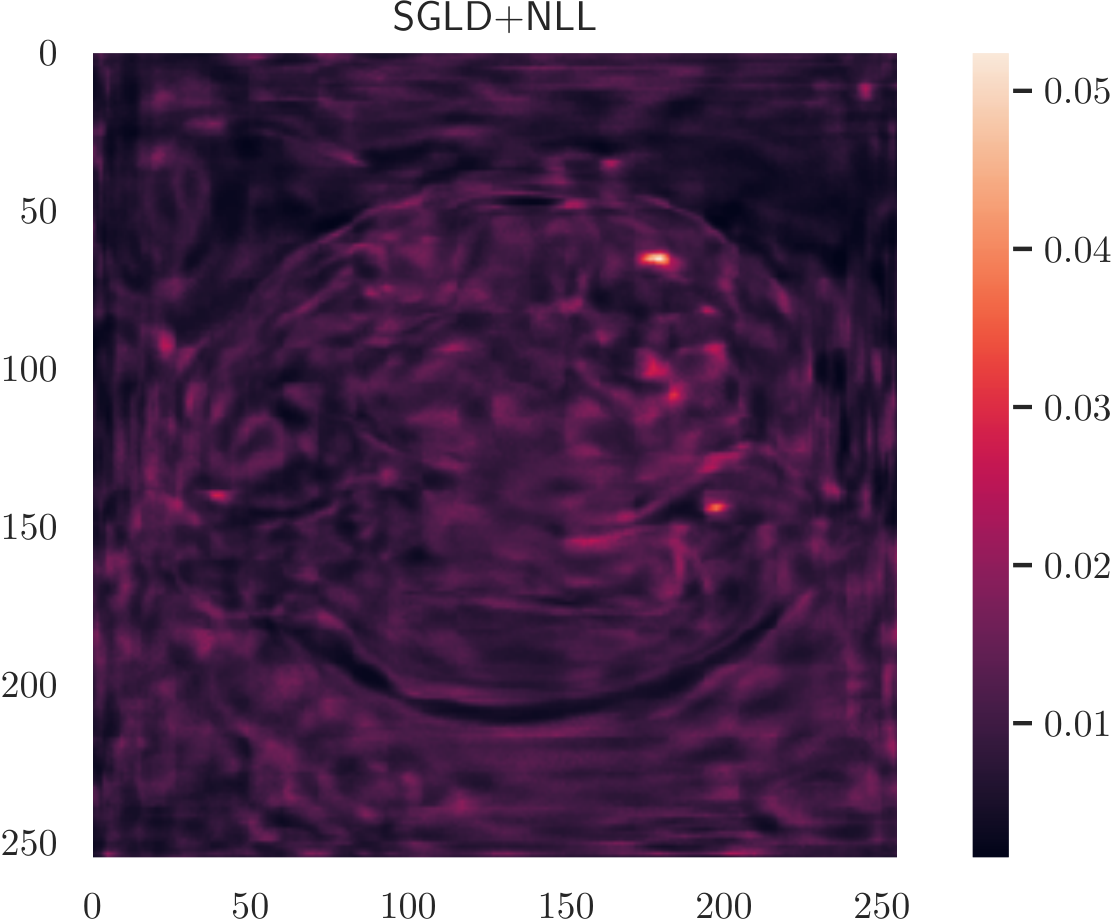}
    \includegraphics[height=2.3cm]{uncert_map_US_mcdip-crop.pdf}
    \caption{Calibration diagrams and uncertainty maps for SGLD+NLL after early stopping and MCDIP after convergence (best viewed with digital zoom). (Left) The calibration diagrams show MSE vs.\ uncertainty and provide mean uncertainty (U) and UCE values. (Right) Uncertainty maps show per-pixel uncertainty.}
\end{figure}

\subsection{Additional Tables}
\label{app:tables}
\enlargethispage{1\baselineskip}

\begingroup
\setlength{\tabcolsep}{6pt}
\begin{table}[!h]
    \centering
    \caption{PSNR with early-stopping.}
    \label{tab:psnr_early_stopping}
    \begin{tabular}{rcccc}
        \toprule
        PSNR   & DIP & SGLD & SGLD+NLL & MCDIP \\
        \cmidrule{2-5}
        OCT    & $ 29.88 \pm 0.02 $ & $ 29.89 \pm 0.05 $ & $ 29.77 \pm 0.07 $ & $ \mathbf{29.92} \pm 0.03 $ \\
        US     & $ 29.74 \pm 0.05 $ & $ \mathbf{29.78} \pm 0.02 $ & $ 29.54 \pm 0.03 $ & $ 29.7 \pm 0.07 $ \\
        X-ray  & $ 30.91 \pm 0.05 $ & $ 30.98 \pm 0.09 $ & $ 30.74 \pm 0.03 $ & $ \mathbf{31.22} \pm 0.1 $ \\
        \bottomrule
    \end{tabular}
\end{table}

\begin{table}[!h]
    \centering
    \caption{SSIM after convergence.}
    \label{tab:ssim}
    \begin{tabular}{rcccc}
        \toprule
        SSIM   & DIP & SGLD & SGLD+NLL & MCDIP \\
        \cmidrule{2-5}
        OCT    & $ 0.582 \pm 0.0 $ & $ 0.574 \pm 0.0 $ & $ 0.66 \pm 0.0 $ & $ \mathbf{0.872} \pm 0.0 $ \\
        US     & $ 0.687 \pm 0.0 $ & $ 0.703 \pm 0.0 $ & $ 0.723 \pm 0.0 $ & $ \mathbf{0.902} \pm 0.0 $ \\
        X-ray  & $ 0.625 \pm 0.0 $ & $ 0.631 \pm 0.0 $ & $ 0.686 \pm 0.0 $ & $ \mathbf{0.922} \pm 0.0 $ \\
        \bottomrule
    \end{tabular}
\end{table}

\begin{table}[!h]
    \centering
    \caption{SSIM with early-stopping.}
    \label{tab:ssim_early_stopping}
    \begin{tabular}{rcccc}
        \toprule
        SSIM   & DIP & SGLD & SGLD+NLL & MCDIP \\
        \cmidrule{2-5}
        OCT    & $ 0.872 \pm 0.0 $ & $ 0.872 \pm 0.0 $ & $ 0.872 \pm 0.0 $ & $ 0.872 \pm 0.0 $ \\
        US     & $ 0.902 \pm 0.0 $ & $ \mathbf{0.903} \pm 0.0 $ & $ 0.899 \pm 0.0 $ & $ \mathbf{0.903} \pm 0.0 $ \\
        X-ray  & $ 0.915 \pm 0.0 $ & $ 0.917 \pm 0.0 $ & $ 0.912 \pm 0.0 $ & $ \mathbf{0.923} \pm 0.0 $ \\
        \bottomrule
    \end{tabular}
\end{table}
\endgroup

\subsection{SGLD With Step Size Decay}
\label{app:sgldlr}

Additionall, we implement SGLD with step size decay as described by Welling et al.\ \cite{Welling2011}.
The step size $ \epsilon $ is used to scale the parameter update in the SGD step (i.e.\ the learning rate) and defines the variance of the noise that is injected into the gradients.
Here, we reduce the step size at each step $ t $ exponentially with $ \epsilon_{t} = 0.999^{t} \epsilon_{0} $.
To satisfy the step size property (Eq.\,(2) in \cite{Welling2011}), we fix the step size once it decreases below 1e-8.
We observe no overfitting of the noisy image with step size decay (see Fig.\,\ref{fig:sgldlr}).
However, the quality of the resulting denoised image is very sensitive to the decay scheme.
Choosing a decrease that is too low (i.e.\ $ \epsilon_{t} = 0.9999^{t} \epsilon_{0} $) results in overfitting; a decrease that is too high  (i.e.\ $ \epsilon_{t} = 0.99^{t} \epsilon_{0} $) results in convergence to a subpar reconstruction.
This is equivalent to carefully applied early stopping and therefore nullifies the advantage of SGLD for denoising of medical images.

\begin{figure}
    \centering
    \includegraphics[scale=0.48]{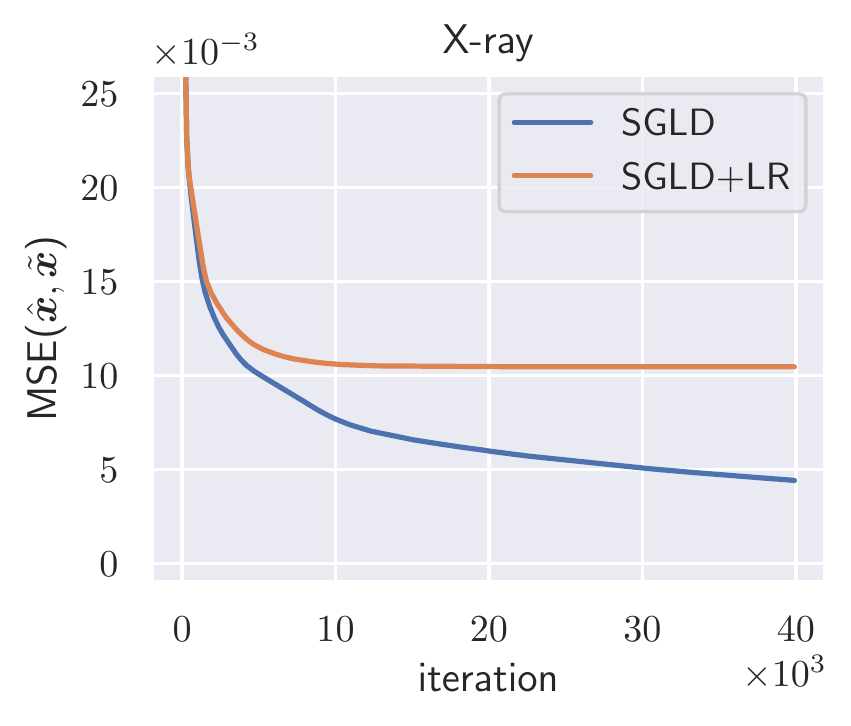}
    \includegraphics[scale=0.48]{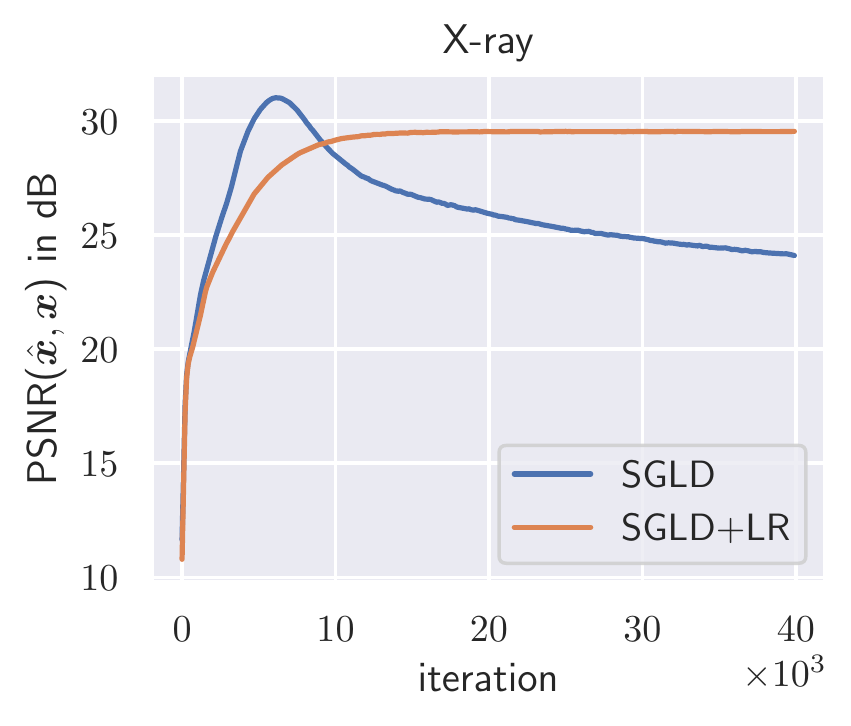} ~
    \includegraphics[height=3.4cm]{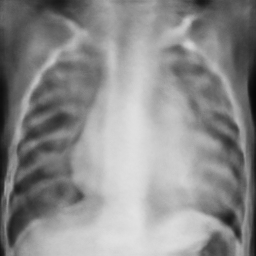}
    \caption{Comparison of SGLD and SGLD+LR (with step size decay). Carefully chosen step size decay impedes overfitting the noisy image. (Right) Reconstruction of SGLD+LR after convergence (no early stopping applied).}
    \label{fig:sgldlr}
\end{figure}

\subsection{Downsampling}
\label{app:downsampling}

Here, we provide justification why downsampling of an image by averaging neighboring pixels reduces the noise level and can be used as an approximation to a ground truth noise-free image (by sacrificing image resolution).

\begin{proposition}
    Downsampling of an image reduces the observation noise.
\end{proposition}

\begin{proof}
    Let $ X = \mu_{x} + \varepsilon_{x} $ and $ Y = \mu_{y} + \varepsilon_{y} $ be two neighboring pixels affected by additive i.i.d.\ noise $ \varepsilon_{x} , \varepsilon_{y} \sim \mathcal{N}(0, \sigma^{2}) $.
    The pixels are assumed to be uncorrelated to noise.
    Pixels in a local neighborhood are highly correlated and assumed to be of high similarity $ \mu_{x} \approx \mu_{y} = \mu $.
    Let $ Z = \tfrac{1}{2} \left( X + Y \right) $ be the average of two neighboring pixels (i.e.\ the result of downsampling).
    The expectation is given by
    \begin{align}
        \mathbb{E}[Z] &= \frac{1}{2} \left( \mathbb{E}[X] + \mathbb{E}[Y] \right) \\
        &= \frac{1}{2} 2 \, \mathbb{E}[X] \\
        &= \mu
    \end{align}
    and the variance is given by
    \begin{align}
        \mathrm{Var}\left[Z\right] &= \mathrm{Var}\left[ \frac{1}{2} \left( X + Y \right) \right] \\
        &= \frac{1}{2^{2}} \left( \mathrm{Var}\left[X\right] + \mathrm{Var}\left[Y\right] \right) \\
        &= \frac{1}{2^{2}} 2 \mathrm{Var}\left[X\right] \\
        &= \frac{1}{2} \sigma^{2} ~ .
    \end{align}
    Thus, if the similarity of neighboring pixels is sufficiently high,
    downsampling reduces the variance of average pixel $ Z $ by a factor of $ 2 $. \qed
\end{proof}

Naturally, two neighboring pixels are not exactly equal.
However, downsampling can also be viewed as superposing two signals, each with a highly correlated and an uncorrelated part.
Without providing proof, the amplitude of the addition of two signals can be viewed as vector addition.
In the uncorrelated case, the two signals are perpendicular to each other and in the correlated case, the angle between the two signals is acute.
Thus, the correlated parts of the two signals have a higher impact on the resulting addition than the uncorrelated (noise) parts.
In the ideal case, where the noise is uncorrelated and the signals are in parallel, the same noise reduction as above follows.

\subsection{Link Between Poisson Distribution and Normal Distribution}
\label{app:poisson}

We approximate the Poisson noise to simulate a low-dose X-ray image with a Normal distribution.
It is well-known that the limiting distribution of $ \mathsf{Poisson}(\lambda) $ is Normal as $ \lambda \rightarrow \infty $ \cite{Hogg2018}.
For completeness, we list a common proof using the moment generating function of a standardized Poisson random variable:

\begin{theorem}
    The \textsf{Poisson}($\lambda$) distribution can be approximated with a Normal distribution as $ \lambda \rightarrow \infty $.
\end{theorem}

\begin{proof}
    Let $ X_{\lambda} \sim \mathsf{Poisson}(\lambda), ~ \lambda \in \{ 1, 2, \ldots \} $.
    The probability mass function of $ X_{\lambda} $ is given by
    \begin{equation}
        f_{X_{\lambda}}(x) = \frac{\lambda^{x}e^{-\lambda}}{x!} \quad x \in \{ 0, 1, 2, \ldots \} ~ .
    \end{equation}
    The moment generating function is given by \cite{Hogg2018}
    \begin{equation}
        M_{X_{\lambda}}(t) = \mathbb{E} [ e^{t X_{\lambda}} ] = e^{\lambda (e^{t}-1)} ~ .
    \end{equation}
    The standardized Poisson random variable
    \begin{equation}
        Z = \frac{X_{\lambda} - \lambda}{\sqrt{\lambda}}
    \end{equation}
    has the limiting moment generating function
    \begin{align}
       \lim_{\lambda \rightarrow \infty} M_{Z} (t) &= \lim_{\lambda \rightarrow \infty} \mathbb{E} \left[ \exp{\left( t \cdot \frac{X_{\lambda} - \lambda}{\sqrt{\lambda}} \right)} \right] \\
        &= \lim_{\lambda \rightarrow \infty} \exp{ \left( -t \sqrt{\lambda}  \right) } \mathbb{E} \left[ \exp{\left( \frac{t X_{\lambda}}{\sqrt{\lambda}} \right)} \right] \\
        &= \lim_{\lambda \rightarrow \infty} \exp{ \left( -t \sqrt{\lambda}  \right) } \exp{\left( \lambda \left( e^{t/\sqrt{\lambda}} - 1 \right) \right)} \\
        &= \lim_{\lambda \rightarrow \infty} \exp{ \left( -t \sqrt{\lambda}  + \lambda \left( t \lambda^{-1/2} + t^{2} \lambda^{-1}/2 + t^{3} \lambda^{-3/2}/6 + \ldots \right) \right) } \\
        &= \lim_{\lambda \rightarrow \infty} \exp{ \left( t^{2} / 2 + t^{3}\lambda ^{-1/2}/6 + \ldots \right) } \\
        &= \exp{\left( t^{2} / 2 \right)}
    \end{align}
    which is the moment generating function of a standard normal random variable.
    \qed
\end{proof}

\end{document}